\documentclass[11pt,a4paper]{article}
\usepackage[top=3cm, bottom=2.5cm, right=2.5cm, left=2.5cm]{geometry}

\usepackage{moreverb}
\usepackage{booktabs}

\usepackage[colorlinks,bookmarksopen,bookmarksnumbered,citecolor=red,urlcolor=red]{hyperref} 

\usepackage[utf8]{inputenc}
\usepackage{amsmath}
\usepackage{amsthm}
\usepackage{amssymb}
\usepackage{mathrsfs}
\usepackage{amsfonts}
\usepackage{bbold}

\usepackage{algorithm}
\usepackage{algorithmicx}
\usepackage{algpseudocode}
\usepackage{graphicx}
\usepackage[normalsize]{subfigure}
\usepackage{xifthen}
\usepackage{rotating}

\newcommand\BibTeX{{\rmfamily B\kern-.05em \textsc{i\kern-.025em b}\kern-.08em
T\kern-.1667em\lower.7ex\hbox{E}\kern-.125emX}}

\newcommand{\dlt}[1]{}
\newcommand{\JV}[1]{#1}
\newcommand{\JVe}[1]{#1}

\newcommand{\skipifemptyarg}[1]{\ifthenelse{\isempty{#1}}{}{\left[#1\right]}}
\newcommand{\skipifscalar}[1]{\ifthenelse{\isempty{#1}}{}{;#1}}

\newcommand{\bs}[1]{\boldsymbol{#1}}
\newcommand{\scal}[2]{\bigl(#1,#2\bigr)}
\newcommand{\bildual}{^{-1}}
\newcommand{\bilf}[2]{a\ifthenelse{\isempty{#1}}{}{\bigl(#1,#2\bigr)}}
\newcommand{\bilfi}[2]{a\bildual\ifthenelse{\isempty{#1}}{}{\bigl(#1,#2\bigr)}}
\newcommand{\bilfG}[2]{\mathring{a}\ifthenelse{\isempty{#1}}{}{\bigl(#1,#2\bigr)}}
\newcommand{\bilfGi}[2]{\mathring{a}\bildual\ifthenelse{\isempty{#1}}{}{\bigl(#1,#2\bigr)}}
\newcommand{\bilfN}[2]{\widetilde{a}_\VN\ifthenelse{\isempty{#1}}{}{\bigl(#1,#2\bigr)}}
\newcommand{\bilfNi}[2]{\widetilde{a}_\VN\bildual\ifthenelse{\isempty{#1}}{}{\bigl(#1,#2\bigr)}}
\newcommand{\bilfD}[3]{\M{a}_{#1}\ifthenelse{\isempty{#2}}{}{\bigl(#2,#3\bigr)}}
\newcommand{\bilfDi}[3]{\M{a}\bildual_{#1}\ifthenelse{\isempty{#2}}{}{\bigl(#2,#3\bigr)}}
\newcommand{\bilfDF}[3]{\widehat{\M{a}}_{#1}\ifthenelse{\isempty{#2}}{}{\bigl(#2,#3\bigr)}}
\newcommand{\bilfDFi}[3]{\widehat{\M{a}}\bildual_{#1}\ifthenelse{\isempty{#2}}{}{\bigl(#2,#3\bigr)}}
\newcommand{\bilfDN}[2]{\tilde{\M{a}}_\VN\ifthenelse{\isempty{#1}}{}{\bigl(#1,#2\bigr)}}
\newcommand{\bilfDNi}[2]{\tilde{\M{a}}_\VN\bildual\ifthenelse{\isempty{#1}}{}{\bigl(#1,#2\bigr)}}
\newcommand{\obilf}[2]{\ol{a}\ifthenelse{\isempty{#1}}{}{\bigl(#1,#2\bigr)}}
\newcommand{\ubilfi}[2]{\ul{a}\bildual\ifthenelse{\isempty{#1}}{}{\bigl(#1,#2\bigr)}}
\newcommand{\set}[1]{\mathbb{#1}}
\newcommand{\M}[1]{\mat{#1}}
\newcommand{\MB}[1]{\boldsymbol{\mat{#1}}}
\newcommand{\mat}[1]{\mathsf{#1}} 
\newcommand{\V}[1]{\bs{#1}}
\newcommand{\T}[1]{\bs{#1}}
\newcommand{\ol}[1]{\overline{#1}}
\newcommand{\ul}[1]{\underline{#1}}

\newcommand{\mac}[1]{^{(#1)}}
\newcommand{\sub}[1]{^{#1}}
\newcommand{\incl}[1]{_{(#1)}}
\newcommand{\iter}[1]{_{(#1)}}

\newcommand{\puc}{\mathcal{Y}}
\newcommand{\pVN}{\meas{\VN}}
\newcommand{\ptVN}{\meas{\tVN}}
\newcommand{\ptVNr}{\meas{\tVNr}}
\newcommand{\pVM}{\meas{\VM}}
\newcommand{\pVP}{\meas{\VP}}
\newcommand{\per}{\#}
\newcommand{\eff}{{\mathrm{H}}}
\newcommand{\full}{{\mathrm{full}}}

\newcommand{\dime}{d}
\newcommand{\bound}{^{\mathrm{bound}}}

\newcommand{\Vf}{\V{f}}
\newcommand{\imu}{\mathrm{i}}		

\newcommand{\cA}{c_A}
\newcommand{\CA}{C_A}
\newcommand{\rA}{\rho_A}

\newcommand{\Aeff}{\TA_{\eff}}
\newcommand{\Beff}{\TB_{\eff}}
\newcommand{\AeffN}{\TA_{\eff,\VN}}
\newcommand{\BeffN}{\TB_{\eff,\VN}}

\newcommand{\cb}[1]{\V{U}^{(#1)}}

\newcommand{\del}{\ensuremath{\delta}}
\newcommand{\alp}{\ensuremath{\alpha}}
\newcommand{\bet}{\ensuremath{\beta}}

\newcommand{\xR}{\set{R}}

\newcommand{\xZ}{\set{Z}}
\newcommand{\xC}{\set{C}}
\newcommand{\xRd}{{\set{R}^{d}}}
\newcommand{\xCd}{\set{C}^{d}}
\newcommand{\xCdN}{\set{C}^{d\times\VN}}
\newcommand{\xRdd}{\set{R}^{d\times d}}
\newcommand{\xCdd}{\set{C}^{d\times d}}
\newcommand{\xRdN}{\xR^{d\times \VN}}
\newcommand{\xRdM}{\xR^{d\times \VM}}
\newcommand{\xCdM}{\xC^{d\times \VM}}
\newcommand{\xRdtN}{{\xR^{d\times\tVN}}}

\newcommand{\xRddspd}{\set{R}_{\mathrm{spd}}^{d\times d}}
\newcommand{\xMN}{\bigl[\set{R}^{d\times\VN}\bigr]^2}
\newcommand{\xMM}{\bigl[\set{R}^{d\times\VM}\bigr]^2}
\newcommand{\xMtN}{\bigl[\set{R}^{d\times\tVN}\bigr]^2}

\newcommand{\xhMN}{\bigl[\set{C}^{d\times\VN}\bigr]^2}

\newcommand{\xXN}{\xRdN}
\newcommand{\xXM}{\xRdM}
\newcommand{\xhXM}{\xCdM}

\newcommand{\xhXN}{\set{C}^{d\times\VN}}
\newcommand{\xX}{\set{X}}
\newcommand{\hX}{\hat{\set{X}}}
\newcommand{\cH}{\mathscr{H}}

\newcommand{\xNd}{\set{N}^d}
\newcommand{\Zd}{\set{Z}^d}
\newcommand{\ZdmO}{\set{Z}^d \backslash \{\T{0}\}}

\newcommand{\ZtNd}{\set{Z}^d_{\tVN}}
\newcommand{\ZtNrd}{\set{Z}^d_{\tVNr}}
\newcommand{\ZNd}{\set{Z}^d_{\VN}}
\newcommand{\ZMd}{\xZ^d_{\V{M}}}
\newcommand{\ZPd}{\xZ^d_{\V{P}}}
\newcommand{\xS}{\set{S}}
\newcommand{\xSd}{\set{S}^d}

\newcommand{\cU}{\mathscr{U}}
\newcommand{\cE}{\mathscr{E}}
\newcommand{\cJ}{\mathscr{J}}
\newcommand{\rcU}{\mathring{\mathscr{U}}}
\newcommand{\rcE}{\mathring{\mathscr{E}}}

\newcommand{\cEN}{\cE_\VN}
\newcommand{\cJN}{\cJ_\VN}

\newcommand{\Lper}[2]{L^{#1}_\#(\puc\skipifscalar{#2})}
\newcommand{\Ltp}{L^{2}_{\per}}
\newcommand{\Cper}[2]{C^{#1}_\#(\puc\skipifscalar{#2})}

\newcommand{\FT}[1]{\widehat{#1}}

\newcommand{\xUN}{\set{U}_\VN}
\newcommand{\xEN}{\set{E}_\VN}
\newcommand{\xJN}{\set{J}_\VN}

\newcommand{\xUNM}{\set{U}_{\VM}}
\newcommand{\xENM}{\set{E}_{\VN,\VM}}
\newcommand{\xJNM}{\set{J}_{\VN,\VM}}
\newcommand{\xUNN}{\set{U}_{\VN}}
\newcommand{\xENN}{\set{E}_{\VN,\VN}}
\newcommand{\xJNN}{\set{J}_{\VN,\VN}}
\newcommand{\xUNF}{\widehat{\set{U}}_\VN}
\newcommand{\xENF}{\widehat{\set{E}}_\VN}
\newcommand{\xJNF}{\widehat{\set{J}}_\VN}

\newcommand{\cT}{\mathscr{T}}
\newcommand{\cTN}{\cT_\VN}

\newcommand{\cTtNrd}{\cT^d_\tVNr}

\newcommand{\cTNd}{\cT_\VN^d}

\newcommand{\TA}{\ensuremath{\T{A}}}
\newcommand{\TB}{\ensuremath{\T{B}}}

\newcommand{\mB}{\ensuremath{\T{B}}}
\newcommand{\rTA}{\ensuremath{\mathring{\T{A}}}}

\newcommand{\mI}{\ensuremath{{\T{I}}}}
\newcommand{\mM}{\ensuremath{\T{M}}}

\newcommand{\mhG}{\ensuremath{\T{\hat{\Gamma}}}}

\newcommand{\VN}{\ensuremath{{\V{N}}}}
\newcommand{\tVNr}{\ensuremath{{2\VN-\V{1}}}}
\newcommand{\tVN}{\ensuremath{{2\VN}}}
\newcommand{\VM}{\ensuremath{{\V{M}}}}
\newcommand{\VP}{\ensuremath{{\V{P}}}}
\newcommand{\Ve}{\ensuremath{\V{e}}}
\newcommand{\Vj}{\ensuremath{\V{\jmath}}}
\newcommand{\rVe}{\ensuremath{\mathring{\Ve}}}
\newcommand{\gVe}{\ensuremath{\widetilde{\Ve}}}
\newcommand{\gVj}{\ensuremath{\widetilde{\Vj}}}

\newcommand{\VE}{\ensuremath{\V{E}}}
\newcommand{\VJ}{\ensuremath{\V{J}}}
\newcommand{\Vu}{\ensuremath{\V{u}}}
\newcommand{\Vv}{\ensuremath{\V{v}}}

\newcommand{\Vx}{\ensuremath{\V{x}}}

\newcommand{\Vm}{\ensuremath{{\V{m}}}}
\newcommand{\Vn}{\ensuremath{{\V{n}}}}
\newcommand{\Vk}{\ensuremath{{\V{k}}}}
\newcommand{\Vp}{\ensuremath{{\V{p}}}}

\newcommand{\Vl}{\ensuremath{{\V{l}}}}
\newcommand{\Vr}{\ensuremath{\V{r}}}
\newcommand{\Vs}{\ensuremath{\V{s}}}
\newcommand{\Vh}{{\ensuremath{\V{h}}}}
\newcommand{\Vo}{\ensuremath{{\V{0}}}}

\newcommand{\MBe}{\ensuremath{\MB{e}}}
\newcommand{\rMBe}{\ensuremath{\mathring{\MB{e}}}}
\newcommand{\MBj}{\ensuremath{\MB{j}}}
\newcommand{\gMBe}{\ensuremath{\widetilde{\MB{e}}}}
\newcommand{\gMBj}{\ensuremath{\widetilde{\MB{j}}}}
\newcommand{\MBA}{\ensuremath{\MB{A}}}
\newcommand{\MBhA}{\ensuremath{\MB{\FT{A}}}}
\newcommand{\MBtA}{\ensuremath{\MB{\widetilde{A}}}}

\newcommand{\MBC}{\ensuremath{\MB{C}}}

\newcommand{\MA}{\ensuremath{\M{A}}}

\newcommand{\Mu}{\ensuremath{\M{u}}}
\newcommand{\Mv}{\ensuremath{\M{v}}}
\newcommand{\MBu}{\ensuremath{\MB{u}}}
\newcommand{\MBv}{\ensuremath{\MB{v}}}
\newcommand{\MBhu}{\ensuremath{\MB{\FT{u}}}}
\newcommand{\MBhv}{\ensuremath{\MB{\FT{v}}}}
\newcommand{\MBx}{\ensuremath{\MB{x}}}
\newcommand{\Mx}{\ensuremath{\M{x}}}
\newcommand{\MBy}{\ensuremath{\MB{y}}}
\newcommand{\My}{\ensuremath{\M{y}}}

\newcommand{\MBB}{\ensuremath{\MB{B}}}
\newcommand{\MBS}{\ensuremath{\MB{S}}}
\newcommand{\bMBu}{\ensuremath{\breve{\MB{u}}}}
\newcommand{\bMBv}{\ensuremath{\breve{\MB{v}}}}
\newcommand{\bMu}{\ensuremath{\breve{\M{u}}}}
\newcommand{\bMv}{\ensuremath{\breve{\M{v}}}}

\newcommand{\MBG}{\MB{G}}
\newcommand{\MBhG}{\ensuremath{\MB{\widehat{G}}}}

\newcommand{\MBF}{\ensuremath{\MB{F}}}
\newcommand{\MBFi}{\ensuremath{\MB{F}^{-1}}}

\newcommand{\MBGNM}[1]{\ensuremath{\MBG_{\VN,\VM}^{#1}}}
\newcommand{\MBGNN}[1]{\ensuremath{\MBG_{\VN,\VN}^{#1}}}
\newcommand{\MBhGNM}[1]{\ensuremath{\MBhG_{\VN,\VM}^{#1}}}
\newcommand{\MBhGNN}[1]{\ensuremath{\MBhG_{\VN,\VN}^{#1}}}

\DeclareMathOperator{\rect}{rect}
\DeclareMathOperator{\tri}{tri}
\DeclareMathOperator{\eig}{eig}
\DeclareMathOperator*{\argmin}{arg\,min}

\DeclareMathOperator{\sinc}{sinc}
\DeclareMathOperator{\circl}{circ}

\DeclareMathOperator{\tr}{tr}
\DeclareMathOperator{\curl}{curl}
\let\div\undefined
\DeclareMathOperator{\div}{div}

\newcommand{\IM}[1]{\mathcal{I}_{#1}}
\newcommand{\IMi}[1]{\IM{#1}^{-1}}

\DeclareMathOperator{\esssup}{ess\,sup}
\newcommand{\norm}[2]{\bigl\| #1 \bigr\|_{#2}}
\newcommand{\meas}[1]{|#1|}
\newcommand{\mean}[1]{\langle#1\rangle}
\newcommand{\zmean}{0}
\newcommand{\D}[1]{\,{\mathrm d}#1}
\newcommand{\conj}[1]{\overline{#1}}

\newcommand{\gani}{\mathrm{GaNi}}

\newcommand{\tAeffN}{\widetilde{\TA}_{\eff,\VN}}
\newcommand{\tBeffN}{\widetilde{\TB}_{\eff,\VN}}

\theoremstyle{plain}
\newtheorem{theorem}{Theorem}
\newtheorem{definition}[theorem]{Definition}
\newtheorem{lemma}[theorem]{Lemma}
\newtheorem{corollary}[theorem]{Corollary}
\newtheorem{remark}[theorem]{Remark}

\newtheorem{proposition}[theorem]{Proposition}
\newtheorem{problem}[theorem]{Problem}

\title{Improved guaranteed computable bounds on homogenized properties of periodic media by Fourier-Galerkin method with exact integration}
\usepackage{authblk}
\author[1,2]{Jaroslav~Vond\v{r}ejc
}
\affil[1]{New Technologies for the Information Society, Faculty of Applied Sciences, University of West Bohemia, Univerzitn\'{i} 2732/8, 306 14 Plze\v{n}, Czech Republic. E-mail: \href{mailto:vondrejc@gmail.com}{\texttt{vondrejc@gmail.com}}}
\affil[2]{Technische Universit\"{a}t Braunschweig, Institute of Scientific Computing, Mühlenpfordstrasse~23, 38106 Braunschweig, Germany}
\date{}

\begin{document}
\maketitle
\begin{abstract}
Moulinec and Suquet introduced FFT-based homogenization in 1994, and twenty years later, their approach is still effective for evaluating the homogenized properties arising from the periodic cell problem.
This paper builds on the author's (2013) variational reformulation approximated by trigonometric polynomials establishing two numerical schemes: Galerkin approximation (Ga) and a version with numerical integration (GaNi).
The latter approach, fully equivalent to the original Moulinec-Suquet algorithm, was used to evaluate guaranteed upper-lower bounds on homogenized coefficients incorporating a closed-form double grid quadrature.
Here, these concepts\JV{, based on the primal and the dual formulations,} are employed for the Ga scheme. \JVe{\JV{For the same computational effort, the Ga outperforms the GaNi with more accurate guaranteed bounds and more predictable numerical behaviors.
Quadrature technique} leading to block-sparse linear systems is extended here to materials defined via high-resolution images in a way which allows for effective treatment using the FFT.}
\JV{Memory demands are reduced by a reformulation of the double to the original grid scheme using FFT shifts.}
\JVe{Minimization of the bounds during iterations of conjugate gradients is effective, particularly when incorporating a solution from a coarser grid.
The methodology presented here for the scalar linear elliptic problem could be extended to more complex frameworks.}

\end{abstract}

\section{Introduction}
\label{sec:introduction}
This paper is devoted to FFT-based homogenization (Fourier-Galerkin method), a numerical method for evaluating homogenized (effective) material coefficients which are essential in multiscale design. This method, which is an alternative to Finite Differences~\cite{Flaherty1973}, Finite Elements~\cite{Guedes1990,Geers2010}, Boundary Elements \cite{Eischen1993BEM,Prochazka2003BEM}, or Fast Multipole Methods~\cite{Greengard2006,Helsing2011effective}, Composite Finite Elements \cite{Hackbusch1997}, X-FEM \cite{Legrain2011}, or the Finite Cell Method \cite{Duster2012FCM_homog}, enables the direct treatment of material coefficients defined via high-resolution images.

The method's effectiveness relies on a Fast Fourier Transform (FFT) that is used for matrix-vector multiplication when solving linear systems; this can be provided in $O(N\log N)$ operations which outperforms most of the existing methods.

\JV{The method's reliability is provided by computable guaranteed upper-lower bounds on homogenized properties \cite{VoZeMa2014GarBounds}, which are based on primal and dual variational formulations \cite{Suquet1982dual} along with a conforming approximation \cite{VoZeMa2014FFTH}.
Surprisingly, this can be elegantly and efficiently treated also with the dual formulation since divergence-free fields can be constrained in the Fourier domain, noticed e.g. by \cite{Nemat-Nasser1993book,Bonnet2007}. 
Generally, the conforming approximation in the dual formulation is more difficult to provide because it requires special basis functions and techniques especially for vector-valued problems such as linearized elasticity.}

Up to now, the drawbacks of FFT-based methods compared to the above mentioned methods included low adaptability originating from the use of regular discretization grids or indirect application for materials with holes. Difficulties also arose when treating various complex (nonlinear, coupled) physical problems since the FFT-based method was originally formulated for the Lippmann-Schwinger integral equation incorporating the Green function derived for a reference medium as a parameter of the method. However, the method has already been applied to e.g. large deformations \cite{Kabel2014LargeDef}, viscoelasticity \cite{Smilauer2010identification}, thermo-elasticity \cite{Vinogradov2008AFFT}, or fracture and damage mechanics \cite{Li2012damage}.

\subsection{FFT-based Homogenization}
\label{sec:fft-based-homogenization}
\JVe{For simplicity and clarity, the methodology is presented here only for a scalar linear elliptic problem describing stationary heat transfer, electric conductivity, or diffusion.}

The model problem consists of an evaluation of homogenized properties $\Aeff\in\xRdd$ that comply with the minimization problem
\begin{align}
\label{eq:intro_HP}
\Aeff \VE\cdot\VE
&= \inf_{u\in H^1_{\per,\zmean}(\puc)} 
\int_{\puc}
\TA(\Vx)[\VE+\nabla u(\Vx)] \cdot [\VE+\nabla u(\Vx)] \D{\Vx}
\end{align}
for an arbitrary vector $\VE\in\xRd$.
The region $\puc=\prod_{\alp=1}^d\bigl(-\frac{1}{2},\frac{1}{2}\bigr)\subset{\xRd}$ accounts for a $d$-dimensional cell where the material coefficients are defined through the bounded, symmetric, and uniformly elliptic $\puc$-periodic matrix function $\TA:\xRd\rightarrow\xRdd$.
The trial space $H^1_{\per,\zmean}(\puc)$ consists of the $\puc$-periodic scalar functions $u:\puc\rightarrow\xR$ with a square integrable gradient and zero mean.

The original FFT-based method was proposed in 1994 by Moulinec and Suquet \cite{Moulinec1994FFT} as a new numerical algorithm for solving the Lippmann-Schwinger equation\JV{, derived from \eqref{eq:intro_HP}}. In this paper, FFT-based methods build upon a variational reformulation by the author and co-workers \cite{Vondrejc2013PhD,VoZeMa2014FFTH,VoZeMa2012LNSC} together with an approximation carried out with
a truncated Fourier series space $\cTNd$
with Fourier basis functions $\varphi_\Vk(\Vx)=\exp(2\pi\imu \Vk\cdot \Vx)$ having bounded frequencies $\Vk\in\ZNd$.
Thus, the reference medium parameter is naturally avoided and two discretization schemes to the cell problem \eqref{eq:intro_HP} are revealed:
Galerkin approximation (Ga)
\begin{subequations}
\label{eq:intro_Ga_both}
\begin{align}\label{eq:intro_Ga}
\TA_{\eff,\VN}\VE\cdot\VE
&= \inf_{u_{\VN}\in \cTN } \int_\puc
\TA(\Vx)[\VE + \nabla u_{\VN}(\Vx)]\cdot [\VE + \nabla u_{\VN}(\Vx)]\D{\Vx},
\end{align}
and Galerkin approximation with numerical integration (GaNi)
\begin{align}\label{eq:intro_GaNi}
\tAeffN^{\gani}\VE\cdot\VE &=
\inf_{u_{\VN}\in \cTN } \sum_{\Vk\in\ZNd} \TA(\Vx_\VN^\Vk)[\VE + \nabla u_{\VN}(\Vx_\VN^\Vk)]\cdot [\VE + \nabla u_{\VN}(\Vx_\VN^\Vk)],
\end{align}
\end{subequations}
which is provided by a trapezoidal (rectangular) rule with integration points $\Vx_\VN^\Vk$ for $\Vk\in\ZNd$ located on a regular grid, introduced in section~\ref{sec:trig_pols}, particularly in~\eqref{eq:grid_points} along with the index set $\ZNd$ in \eqref{eq:index_set}.

\JVe{The discretization of the cell problem \eqref{eq:intro_HP} with trigonometric polynomials $\cTNd$ is a standard numerical approach\JV{, which }has been used several times as the spectral Fourier-Galerkin or collocation method, e.g. \cite{Dykaar1992,Nemat-Nasser1993book,Moulinec1994FFT,Luciano1998,Vainikko2000FSLS,Naess2002,Cai2008}, mostly within the Lippmann-Schwinger equation \cite{Nemat-Nasser1993, Bonnet2007} or in a standard variational setting \cite{Luciano1998}.}
\JV{A special attention is attributed to the utilization of the FFT algorithm in \cite{Dykaar1992} or in \cite{Moulinec1994FFT} leading to the Moulinec and Suquet algorithm that was interpreted as a trigonometric collocation method for the Lippmann-Schwinger equation \cite{ZeVoNoMa2010AFFTH} or as a Fourier-Galerkin method with numerical integration \eqref{eq:intro_GaNi} in \cite[section~5.2]{VoZeMa2014FFTH}.
Bonnet in \cite{Bonnet2007} incorporated not only FFT but also exact inclusion geometries into the Lippmann-Schwinger equation; it was interpreted in \cite[section~4.2]{VoZeMa2014FFTH} as a Fourier-Galerkin method with exact integration \eqref{eq:intro_Ga} in the standard variational framework.}

The theories regarding FFT-based methods --- including discretization, convergence of approximate solutions, and solution of the corresponding linear systems --- have been provided only recently by the author and co-workers \cite{VoZeMa2014FFTH} \JV{in the standard variational setting  or later in \cite{Schneider2014convergence} for rough material coefficients within the Lippmann-Schwinger equation.}

\JV{Another theoretically supported approach by Brisard and Dormieux in \cite{Brisard2010FFT,Brisard2012FFT}} describes the method using Galerkin approximation with piece-wise constant basis functions using a Lippmann-Schwinger integral equation. However, the reference medium, a parameter of the method, influences both the quality of the approximate solutions and the convergence of linear solvers.

Apart from the Fourier-Galerkin formulations \eqref{eq:intro_Ga_both} studied in \cite[sections~4.2 and 4.3]{VoZeMa2014FFTH} and the Brisard and Dormieux approach \cite{Brisard2010FFT,Brisard2012FFT} --- both of which lead to FFT-based schemes --- there are various other modifications and improvements.
In \cite{Monchiet2012polarization}, the authors provided a polarization-based scheme which can handle arbitrary phase contrast (voids and stiff inclusions); this can be also managed by a numerical method based on augmented Lagrangians \cite{Michel2000CMB}.
Recently, Willot et al. in \cite{willot2013fourier,Willot2015} adjusted the integral kernel in the Lippmann-Schwinger equation which led to improved accuracy in approximate solutions, illustrated with a comparison using an analytical solution \cite{craster2001four}. A~possible improvement can also be achieved by smoothing of material coefficients \cite{Merkert2014},\cite[section~8.3]{VoZeMa2014GarBounds}.
Significant attention was granted to improving the linear solvers leading to the accelerated schemes \cite{Eyre1999FNS,Vinogradov2008AFFT} \JV{or to the Krylov subspace methods, such as conjugate gradients \cite{ZeVoNoMa2010AFFTH,Brisard2010FFT}; a comparison can be found in \cite{Moulinec2014comparison,MiVoZe2015jcp}.}

\subsection{Guaranteed bounds on homogenized properties}

The theory of guaranteed bounds on homogenized coefficients has been the subject of many studies in analytical homogenization theories. These techniques employ the primal-dual formulations of the cell problem \eqref{eq:intro_HP} with limited --- and often uncertain --- information about the material coefficients $\TA$. Specific examples include the Voigt \cite{Voigt1910lehrbuch}, Reuss \cite{Reuss1929}, and Hashin-Shtrikman  bounds \cite{Hashin1963variational}; see the monographs \cite{Nemat-Nasser1993book,Cherkaev2000variational,Milton2002TC,Torquato2002random,Dvorak2012micromechanics}
for a more complete overview. Because the bounds rely on limited data, their performance rapidly deteriorates for highly-contrasted media.

Relatively less attention has been given to the computable upper-lower bounds arising from a conforming  approximation to the cell problem \eqref{eq:intro_HP}.
These bounds can be made arbitrarily accurate if the approximate solutions converge to the solution of \eqref{eq:intro_HP}. Moreover, the bounds are guaranteed if they allow for closed-form evaluation.
Dvo\v{r}\'{a}k and Haslinger, to our knowledge, specified the relevant ideas in their work \cite{Dvorak1993master,Haslinger1995optimum}, which applied the approach to the $p$-version of FEM. The application to the $h$-version of FEM occurs independently in \cite{Wieckowski1995DFEM} \JV{for an elasticity} and to Fourier discretization \JV{within the Lippmann-Schwinger equation (Hashin-Shtrikman functional) \cite{Nemat-Nasser1993,Bonnet2007} or the standard variational setting \cite{Luciano1998}}.

The effective evaluation of upper-lower bounds using FFT has been provided recently in \cite{Kabel2012precisebounds} for linear elasticity and later in \cite{Bignonnet2014fft} for permeability.
Both frameworks rely on the Hashin-Shtrikman functional or the Lippmann-Schwinger equation, respectively, discretized with the Brisard--Dormieux method \cite{Brisard2012FFT} assuming piecewise-constant approximations.
However, the evaluation of bounds is based on the summation of an infinite series leading to the loss of guarantee.

The Fourier-Galerkin approaches used here in \eqref{eq:intro_Ga_both} or in \cite{Nemat-Nasser1993,Luciano1998} provide guaranteed bounds on homogenized properties, which is based on closed-form evaluation of corresponding bilinear forms in \eqref{eq:intro_Ga} for material properties that have an analytical expression of Fourier coefficients.
This approach has been already employed in \cite{VoZeMa2014GarBounds} to GaNi \eqref{eq:intro_GaNi}, while here the focus is on its generalization to and comparison with the Ga \eqref{eq:intro_Ga}.
Moreover, these bounds can be made arbitrarily accurate thanks to the convergence analysis of approximate solutions provided by Vond\v{r}ejc et al. \cite{Vondrejc2013PhD,VoZeMa2014FFTH} and improved by Schneider \cite{Schneider2014convergence} to account for rough coefficients.

The guaranteed bounds incorporating trigonometric polynomials were mostly calculated with the Hashin-Shtrikman functional \cite{Nemat-Nasser1993}.
The work  \cite{Luciano1998} equivalent to \eqref{eq:intro_Ga} is improved here by incorporating the FFT algorithm and by using a double grid quadrature, which leads to a sparse structure according to \cite{Vondrejc2013PhD,VoZeMa2014GarBounds}.
In \cite{Monchiet2015}, these ideas have been applied to linear elasticity illustrating that the classical variational formulation used here is always better than the Hashin-Shtrikman formulation.

\subsection{Content of the paper}

This paper is organized as follows.
Notation and preliminaries to the periodic functions, Fourier transform, and Helmholtz decomposition presented in section~\ref{sec:preliminaries} are followed with the continuous homogenization problem in section~\ref{sec:continuous_formulations}.
Then section~\ref{sec:homogenization_and_galerkin_methods} follows with Fourier-Galerkin discretization for both primal and dual formulations leading to guaranteed bounds on homogenized coefficients; the structure between Ga and GaNi is established here.
The methodology for evaluating guaranteed bounds is developed in section~\ref{sec:eval_integrals}. Particularly in section \ref{sec:grid_composite}, the double grid quadrature from \cite{VoZeMa2014GarBounds} is generalized for materials defined via high-resolution images.
In section \ref{sec:reduce_grid}, the numerical scheme on the double grid is reduced to the original grid using shifts of DFT \eqref{eq:DFT}.
In section \ref{sec:approx_bounds}, the material properties without analytical expression of Fourier coefficients are approximated in a way to still obtain the guaranteed bounds on the homogenized properties.
Section~\ref{sec:lin_sys} is dedicated to linear systems of Fourier-Galerkin schemes and the related computational aspects.
Numerical examples in section~\ref{sec:numerical_examples} confirm the theoretical results and provide a numerical comparison between the Ga \eqref{eq:intro_Ga} and GaNi schemes \eqref{eq:intro_GaNi}.

\section{Notation and preliminaries to the cell problem}
\label{sec:preliminaries}

In the sections \ref{sec:vectors_matrices} and \ref{sec:periodic_functions_and_fourier}, the author introduces notation and recalls some useful facts related to matrix analysis and to spaces of periodic functions and the Fourier transform. Section~\ref{sec:helmoholtz} is dedicated to the Helmholtz decomposition of vector-valued periodic functions and its description with orthogonal projections, essential for the duality arguments in both discrete and continuous settings.

\subsection{Vectors and matrices}
\label{sec:vectors_matrices}
In the subsequent section, $d$ is reserved for the dimension of the model problem, assuming $d=2,3$.
To keep the notation compact, $\xX$
abbreviates the space of scalars, vectors, or matrices, i.e. $\xR$, $\xRd$, or $\xRdd$, and $\hX$ is used for their complex counterparts, i.e.~$\xC$, $\xCd$, or $\xC^{d \times d}$. Vectors and matrices are denoted by boldface letters, e.g.
$\Vu,\Vv \in \xRd$ or $\mM \in \xRdd$, with Greek letters used when referring to
their entries; e.g. $\mM = (M_{\alp\beta})_{\alp,\beta=1,\ldots,\dime}$.
Matrix $\mI = \bigl(\del_{\alp\beta}\bigr)_{\alp\beta}$ denotes the identity matrix whereas the symbol $\del_{\alp\beta}$ is reserved for the Kronecker delta, defined as
$\delta_{\alp\beta} = 1$ for $\alp = \beta$ and $\delta_{\alp\beta} = 0$ otherwise. 

As usual, the matrix-vector product $\mM \T{v}$, the matrix-matrix product $\mM\T{L}$, dot product $\T{u} \cdot \T{v}$, and
the outer product $\T{u} \otimes \T{v}$ refer to 
\begin{align*}
(\mM \Vu)_\alp &= \sum_\beta M_{\alp\beta} u_\beta, 
&
(\mM\T{L})_{\alp\beta} &= \sum_{\gamma} M_{\alp\gamma} L_{\gamma\beta},
&
\T{u} \cdot \T{v}
&=
\sum_\alp u_\alp v_\alp,
&
(\T{u} \otimes \T{v})_{\alp\beta}
&=
u_\alp v_\beta,
\end{align*}
where the author assumes that $\alp$ and $\beta$ range from $1$ to $d$ for the sake of brevity.
Moreover, the spaces are endowed with the following inner product and norms, e.g. 
\begin{align*}
\|\T{u}\|_{\xCd}^2
&= \sum_\alp u_\alp
\overline{u_\alp},
&
\|\mM\|_{\xCdd} &= \max_{\Vu\neq\T{0}} \frac{\|\mM\T{u}\|_{\xCd}}{\|\T{u}\|_{\xCd}}.
\end{align*}

\subsection{Periodic functions and Fourier transform}
\label{sec:periodic_functions_and_fourier}
For a unit cell $\puc = \prod_\alp 
\bigl(-\frac{1}{2},
\frac{1}{2}\bigr)$, a function $\Vf : \xRd \rightarrow \xX$ is 
$\puc$-periodic if
$\Vf(\Vx + \Vk)
= \Vf(\Vx)$ for all $\Vx \in \puc$
and all $\Vk \in \Zd$.
According to \cite{Rudin1986real,Jikov1994HDOIF,SaVa2000PIaPDE},
\begin{align*}
\Lper{p}{\xX}
=
\left\{
\Vf :\puc\rightarrow\xX
: 
\Vf \text{ is $\puc$-periodic, measurable, and }\|\Vf\|_{\Lper{p}{\xX}}<\infty
\right\}
\quad\text{for }p \in \{2,\infty\}
\end{align*}
denotes the Lebesgue space equipped with the norm
\begin{align*}
\norm{\Vf}{\Lper{p}{\xX}}
=
\begin{cases}
\esssup_{\Vx\in\puc} \norm{\Vf(\Vx)}{\xX} & \text{for } p=\infty,
\\
\left( \displaystyle
\int_\puc
\norm{\Vf(\Vx)}{\xX}^2 \D{\Vx}
\right)^{1/2} & \text{for }p=2.
\end{cases}
\end{align*}
The space $\Lper{2}{\xRd}$ is also a Hilbert space with an inner product  
\begin{align*}
\scal{\Vu}{\Vv}_{\Lper{2}{\xRd}} &=
\int_{\puc} \Vu(\Vx)\cdot\Vv(\Vx) \D{\Vx}.
\end{align*}
For the sake of brevity, the author writes $\Lper{p}{}$ instead of $\Lper{p}{\xR}$, and often shortens $\Lper{2}{\xRd}$ to $\Ltp$ when referring to the norms and the inner product.

Every function $\Vf \in \Lper{2}{\xX}$ can be expressed using Fourier series
\begin{align*}
\Vf(\Vx) 
=
\sum_{\Vk \in \Zd}
\FT{\Vf}( \Vk )
\varphi_{\Vk}(\Vx)
\end{align*}
where the Fourier basis functions $\varphi_\Vk$ and Fourier coefficients $\FT{f}(\Vk)$ for $\Vk\in\Zd$ are defined by
\begin{align*}
\varphi_{\Vk}(\Vx)
&=
\exp{
\Bigl(
  2\pi\imu
    \Vk\cdot\Vx
\Bigr)}
&
\FT{\Vf}( \Vk )
&=
\int_\puc
\Vf(\Vx)
\varphi_{-\Vk}(\Vx)
\D{\Vx}
\in \FT{\xX}
\quad
\text{for }
\Vx \in \puc 
\text{ and }
\Vk \in \Zd,
\end{align*}
cf.~\cite[pp.~89--91]{Rudin1986real}. The mean value of function $\Vf\in \Lper{2}{\xX}$ over periodic cell $\puc$ is denoted as 
\begin{align*}
\mean{\Vf} = \int_{\puc}\Vf(\Vx)\D{\Vx} = \FT{\Vf}(\Vo) \in \xX
\end{align*}
and corresponds to the zero-frequency Fourier coefficient.

\subsection{Helmholtz decomposition for periodic functions}
\label{sec:helmoholtz}
Operator $\oplus$ denotes the direct sum of mutually orthogonal subspaces, e.g. $\xRd=\cb{1}\oplus\cb{2}\oplus\dotsc\oplus\cb{d}$ for vectors  $\cb{\alp} = (\del_{\alp\beta})_\beta$.
According to the Helmholtz decomposition \cite[pages~6--7]{Jikov1994HDOIF}, $\Lper{2}{\xRd}$ admits an orthogonal decomposition
\begin{align}\label{eq:Helmholtz_L2}
\Lper{2}{\xRd}=\cU\oplus\cE\oplus\cJ
\end{align}
into the subspaces of constant, zero-mean curl-free, and zero-mean divergence free fields
\begin{subequations}\label{eq:subspaces_UEJ}
\begin{align}
\cU &= \{\Vv\in \Lper{2}{\xRd}:\Vv(\Vx) = \mean{\Vv}\text{ for all }\Vx\in\puc\},
\\
\cE &= \{\Vv\in \Lper{2}{\xRd}:\curl\Vv = \Vo,\mean{\Vv}=\Vo\},
\\
\cJ &= \{\Vv\in \Lper{2}{\xRd}:\div\Vv = 0,\mean{\Vv}=\Vo\}.
\end{align}
\end{subequations}
Here, the differential operators $\curl$ and $\div$
are understood in the Fourier sense, so that
\begin{align*}
(\curl \Vu)_{\alp\beta}
& = \sum_{\Vk\in\Zd} 2\pi\imu \bigl(k_\beta\FT{u}_\alp(\Vk) - k_\alp\FT{u}_\beta(\Vk)\bigr) \varphi_\Vk,
&
\div \Vu
& = \sum_{\Vk\in\Zd} 2\pi\imu \Vk\cdot\FT{\Vu}(\Vk) \varphi_\Vk,
\end{align*}
cf.~\cite[pp.~2--3]{Jikov1994HDOIF} and \cite{SaVa2000PIaPDE}. Furthermore, the constant functions from $\cU$ are naturally identified with vectors from $\xRd$.
 
Alternatively, the subspaces arising in the Helmholtz decomposition \eqref{eq:subspaces_UEJ} can be characterized by the orthogonal projections introduced below.
\begin{definition}\label{def:projection}
Let $\mathcal{G}\sub{\cU}$, $\mathcal{G}\sub{\cE}$, and $\mathcal{G}\sub{\cJ}$ denote operators $\Lper{2}{\xRd}\rightarrow \Lper{2}{\xRd}$ defined via
\begin{align*}
\mathcal{G}\sub{\bullet} [\Vv](\Vx) &= 
\sum_{\Vk\in\Zd} \mhG\sub{\bullet}(\Vk) \FT{\Vv}(\Vk) \varphi_{\Vk}(\Vx)\quad\text{for }\bullet\in\{\cU,\cE,\cJ\},
\end{align*}
where the matrices of Fourier coefficients $\mhG\sub{\bullet}(\Vk)\in\xR^{d\times d}$ read
\begin{align*}
\mhG\sub{\cU}(\Vk) &=
\begin{cases}
\mI
\\
\Vo\otimes\Vo
\end{cases}
&
\mhG\sub{\cE}(\Vk) &=
\begin{cases}
\Vo\otimes\Vo
\\
\frac{\Vk\otimes\Vk}{\Vk\cdot\Vk}
\end{cases}
&
\mhG\sub{\cJ} (\Vk) &=
\begin{cases}
\Vo\otimes\Vo,&\text{for }\Vk = \Vo
\\
\mI - \frac{\Vk\otimes\Vk}{\Vk\cdot\Vk}&\text{for }\Vk\in\ZdmO
\end{cases}.
\end{align*}
\end{definition}
\begin{lemma}
\label{lem:Helmholtz_decomp}
Operators $\mathcal{G}\sub{\cU}$, $\mathcal{G}\sub{\cE}$, and $\mathcal{G}\sub{\cJ}$ are mutually orthogonal projections with respect to the inner product on $\Lper{2}{\xRd}$, on $\cU$,$\cE$, and $\cJ$.
\end{lemma}

\subsection{Continuous formulation}
\label{sec:continuous_formulations}
Here and in the next sections, the matrix field $\TA:\puc\rightarrow\xRddspd$ is reserved for material coefficients which  are required to be essentially
bounded, symmetric, and uniformly elliptic:
\begin{align}
\label{eq:A}
\TA \in \Lper{\infty}{\xRddspd},
&&
\cA \norm{\Vv}{\xRd}^2
\leq 
\TA(\Vx)\Vv\cdot\Vv
\leq \CA \norm{\Vv}{\xRd}^2,
\end{align}
almost everywhere in $\puc$ and for all $\Vv \in \xRd$ with  $0 < \cA \leq \CA < +\infty$. By $\rA = \CA / \cA$, the author denotes the contrast in coefficients $\TA$.

Employing material coefficients, bilinear forms $\bilf{}{},\bilfi{}{}:\Lper{2}{\xRd}\times\Lper{2}{\xRd}\rightarrow\xR$ 
are defined as
\begin{align}\label{eq:bilinear_forms}
\bilf{\Vu}{\Vv} &:= \scal{\TA\Vu}{\Vv}_{\Lper{2}{\xRd}}
&
\bilfi{\Vu}{\Vv} &:= \scal{\TA^{-1}\Vu}{\Vv}_{\Lper{2}{\xRd}}.
\end{align}

Next, the author defines homogenization problem \eqref{eq:intro_HP} in both primal and dual formulations.
\begin{definition}[Homogenized matrices]
\label{def:homog_problem}
Let material coefficients satisfy \eqref{eq:A}. Then the primal and the dual homogenized matrices $\Aeff, \Beff \in \xR^{d\times d}$ satisfy
\begin{subequations}
\label{eq:homog_problem}
\begin{align}
\label{eq:homog_problem_primal}
\scal{\Aeff \VE}{\VE}_{\xRd} &= \min_{\Ve\in\cE} \bilf{\VE+\Ve}{\VE+\Ve}=\bilf{\VE+\Ve\mac{\VE}}{\VE+\Ve\mac{\VE}},
\\
\scal{\Beff\VJ}{\VJ}_{\xRd} &= \min_{\Vj\in\cJ} \bilfi{\VJ+\Vj}{\VJ+\Vj}=\bilfi{\VJ+\Vj\mac{\VJ}}{\VJ+\Vj\mac{\VJ}}
\end{align}
\end{subequations}
for arbitrary quantities $\VE,\VJ\in\xRd$ and minimizers $\Ve\mac{\VE}$ and $\Vj\mac{\VJ}$.
\end{definition}
\begin{remark}\label{eq:eqiv2pot}
Notice that the primal formulation \eqref{eq:homog_problem_primal} coincides with the problem \eqref{eq:intro_HP} in the introduction, because the subspace $\cE$ from \eqref{eq:Helmholtz_L2} admits an equivalent characterization
$\cE = \{ \nabla f:f\in H^1_{\per}(\puc,0) \}$.
\end{remark}
\begin{remark}[Observations]Thanks to the boundedness and ellipticity of material coefficients \eqref{eq:A}, the homogenization problems in Definition~\ref{def:homog_problem} are
well-posed with unique minimizers. Homogenized matrices are symmetric, positive definite, and thus invertible in accordance with the periodic homogenization theory, e.g.~\cite{Bensoussan1978per_structures,Jikov1994HDOIF,Cioranescu1999Intro2Homog}. Moreover, they are mutually inverse
\begin{align}\label{eq:Aeff_duality}
\Aeff = \Beff^{-1}
\end{align}
in compliance with standard duality arguments \cite{ekeland1976convex}.
\end{remark}

\section{Fourier-Galerkin discretization}
\label{sec:homogenization_and_galerkin_methods}

This section describes the discretization of the homogenization problem \eqref{eq:homog_problem} utilizing the Galerkin method in section~\ref{sec:Galerkin_approx} with an approximation space consisting of trigonometric polynomials  in section~\ref{sec:trig_pols}. Then, the homogenized matrices defined with discrete problems are studied in section~\ref{sec:bounds} to provide a structure of guaranteed bounds on homogenized properties.

\subsection{Approximation space of trigonometric polynomials}
\label{sec:trig_pols}
\begin{definition}[Trigonometric polynomials]\label{def:trig_pol}
The number
\begin{align}
\label{eq:N_odd}
\VN\in\xNd
\quad\text{such that }
N_\alp
\text{ is odd for all }
\alp
\end{align}
is an order of the approximation space of $\xRd$-valued trigonometric polynomials defined as
\begin{align}
\label{eq:trig_space}
\cTNd &= 
\Bigl\{\sum_{\Vk\in\ZNd}{\MB{\FT{v}}^{\Vk}\varphi_{ \Vk }} : \MB{\FT{v}}^{\Vk} = \overline{\MB{\FT{v}}^{-\Vk}} \in\xCd \Bigr\}
\quad\text{for }
\varphi_{\Vk}(\Vx) = \exp\left(2\pi\imu \Vk\cdot \Vx\right) \text{ with } \Vx\in\puc,
\end{align}
where the frequencies $\Vk$ are confined within the index set 
\begin{align}
\label{eq:index_set}
\ZNd  &= 
  \Bigl\{ \Vk \in \set{Z}^d : 
    -\frac{N_\alpha}{2} \leq k_\alpha < \frac{N_\alpha}{2} \Bigr\}.
\end{align}
\end{definition}
The approximation space \eqref{eq:trig_space} consists of real-valued functions only due to the prescribed Hermitian symmetry of the Fourier coefficients, $\MB{\FT{v}}^{\Vk} = \overline{\MB{\FT{v}}^{-\Vk}}$. Because of this and the necessity of having a conforming approximation space in order to get guaranteed bounds on homogenized properties, 
the highest~(Nyquist) frequencies $k_\alp = -N_\alp/2$ are fully omitted in the definition resulting from the restriction to
the odd order $\VN$ according to \eqref{eq:N_odd}.

Next, I present the crucial property of trigonometric polynomials from space $\cTNd$: they can be uniquely expressed on grid points
\begin{align}\label{eq:grid_points}
\Vx_\VM^\Vm
=
\sum_{\alp}\frac{m_{\alp}}{M_{\alp}}\cb{\alp}\quad
\text{for }
\VM\in\xNd
\text{ such that }M_\alp \geq N_\alp
\text{ and } 
\Vm \in \ZMd,
\end{align}
depicted in Figure~\ref{fig:grid}. This can be shown with Discrete Fourier Transform (DFT) coefficients presented here along with their orthogonality property
\begin{align}
\label{eq:DFT_coef}
\omega_{\VM}^{\Vk\Vn} &=
\exp   \biggl(2 \pi \imu\sum_{\alp} \frac{k_\alp n_\alp}{M_\alp}  \biggr), 
&
\sum_{\Vm\in\ZMd}\omega_{\VM}^{-\Vk\Vm}\omega_{\VM}^{\Vm\Vn} &= \meas{\VM}\del_{\Vk\Vn}\quad\text{for }\Vk,\Vn\in\ZMd.
\end{align}
Therefore, a Fourier series of function $\Vu_\VN\in\cTNd$, noting that $\FT{\Vu}_{\VN}(\Vk)=0$ for $\Vk\in\ZMd\setminus\ZNd$ and $\VM\in\xNd$ such that $M_\alp\geq N_\alp$, can be recast by substituting of the DFT coefficients \eqref{eq:DFT_coef} into
\begin{align}
\label{eq:trig_conn_raw}
\Vu_{\VN}(\Vx) & =  
\sum_{\Vk\in\ZMd}\FT{\Vu}_{\VN}(\Vk)\varphi_{\Vk}(\Vx)
=
\sum_{\Vm\in\ZMd}
\underbrace{\sum_{\Vn\in\ZMd}\omega_\VM^{\Vm\Vn}\FT{\Vu}_{\VN}(\Vn)}_{\Vu_{\VN}(\Vx_\VM^{\Vm})}
\underbrace{\sum_{\Vk\in\ZMd} 
\frac{1}{\meas{\VM}}
\omega_\VM^{-\Vk\Vm}
\varphi_{\Vk}(\Vx)}_{\varphi_{\VM,\Vm}(\Vx)},
\end{align}
which results in a fundamental trigonometric polynomial, see Figure~\ref{fig:trig_basis_shape}, satisfying the Dirac delta property on grid points
\begin{align}
\label{eq:trig_basis_shape}
\varphi_{\VM,\Vm}( \Vx )
&=
\frac{1}{\pVM}
\sum_{\Vk \in \ZMd}
\omega_{\VM}^{-\Vk\Vm} \varphi_{\Vk}(\Vx),
&
\varphi_{\VM,\Vm}(\Vx_\VM^{\Vn}) &= \del_{\Vm\Vn}.
\end{align}
This implies that the inverse DFT of Fourier coefficients $\FT{\Vu}_\VN(\Vn)$ equals the function values on grid points, i.e. $\sum_{\Vn\in\ZMd}\omega_\VM^{\Vm\Vn}\FT{\Vu}_{\VN}(\Vn)= \Vu_{\VN}(\Vx_\VM^{\Vm})$. Thus, the trigonometric polynomial can be uniquely defined with the Fourier coefficients or with the values at grid points.

In the following text, the letter $\VN$ will be systematically used for odd approximation order \eqref{eq:N_odd} of trigonometric polynomials, while $\VM$ will represent the number of grid points \eqref{eq:grid_points} without a restriction to be odd.

\begin{figure}[htp]
\centering
\subfigure[\scriptsize Grid points]{
\includegraphics[scale=0.55]{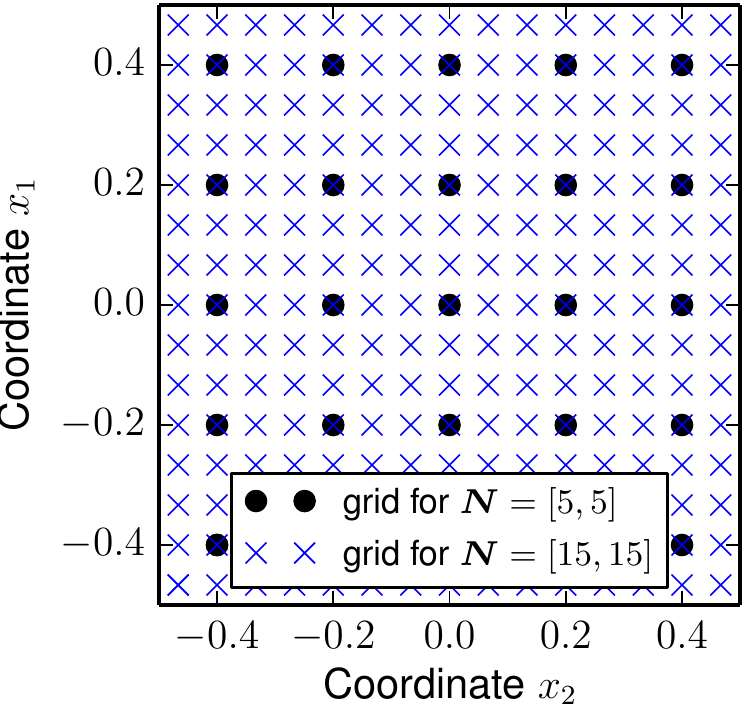}
\label{fig:grid}
}
\subfigure[\scriptsize Trigonometric polynomials]{
\includegraphics[scale=0.55]{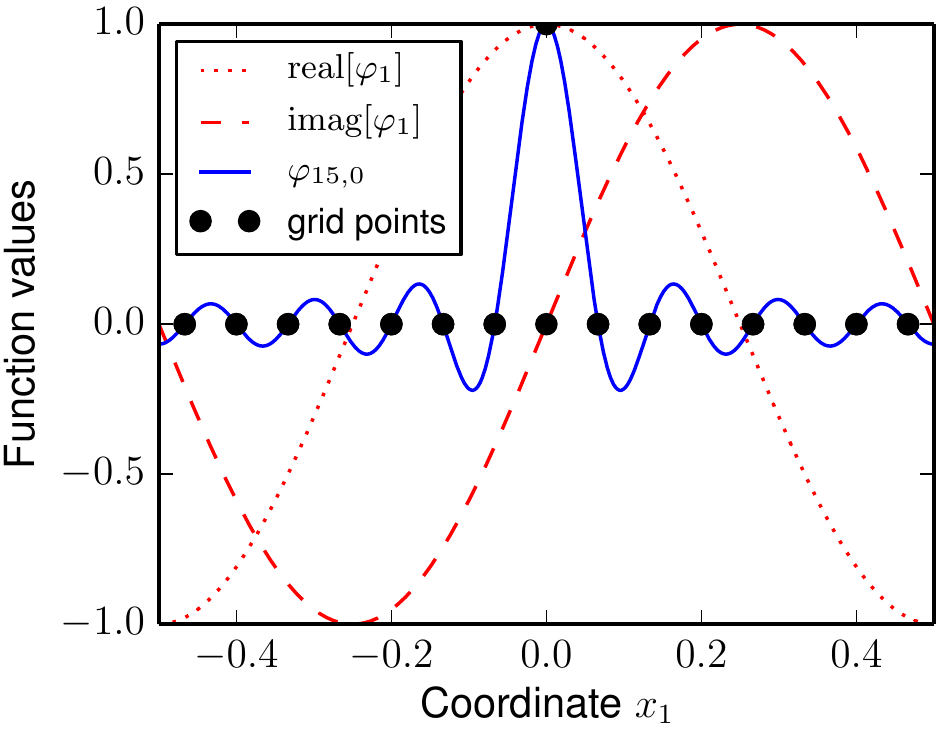}
\label{fig:trig_basis_shape}
}
\caption{Examples of (a) two-dimensional grid points \eqref{eq:grid_points} and (b) one-dimensional complex-valued Fourier $\varphi_1$
and real-valued fundamental $\varphi_{15,0}$ trigonometric polynomials \eqref{eq:trig_basis_shape}
}
\end{figure}

\subsection{Galerkin approximations}
\label{sec:Galerkin_approx}
Generally, the difficulties in a conforming discretization of the continuous homogenization problem \eqref{eq:homog_problem} arise when approximating Helmholtz decomposition subspaces \eqref{eq:Helmholtz_L2}, especially in dual formulations for the space of divergence free fields $\cJ$; the approximation of curl-free space $\cE$ in the primal formulation can be resolved with a conforming approximation of the space of potentials $H^1_{\per,0}(\puc)$.

Luckily, in the case of trigonometric polynomials space \eqref{eq:trig_space}, a Helmholtz decomposition
\begin{align}
\label{eq:Helmholtz_trig}
 \cTNd &= \cU \oplus \cEN \oplus \cJN,
 &
 \cEN&=\cE\bigcap\cTNd,
 &
 \cJN&=\cJ\bigcap\cTNd
\end{align}
can be easily performed with the orthogonal projections $\mathcal{G}\sub{\cU}$, $\mathcal{G}\sub{\cE}$, $\mathcal{G}\sub{\cJ}$ introduced in Definition~\ref{def:projection} for the continuous problem;
the procedure is based on Lemma~\ref{lem:Helmholtz_decomp} and the property $\mathcal{G}\sub{\bullet}[\cTNd] \subset\cTNd$ observed directly from Definition~\ref{def:projection}.
Moreover, these projections enable not only a proper discretization but also an effective numerical treatment of discrete problems, cf.~sections~\ref{sec:discrete_sp} and \ref{sec:linear_system}.

\JVe{Then, the discrete homogenized problems are represented with two schemes, the Galerkin approximation (Ga) and by its version with numerical integration (GaNi) corresponding to the original Moulinec and Suquet scheme \cite{Moulinec1994FFT}. These Fourier-Galerkin schemes occurring in
\cite{Dykaar1992,Nemat-Nasser1993,Luciano1998,Vainikko2000FSLS,Naess2002,Bonnet2007,Cai2008} as spectral methods  were studied in the variational setting in \cite[sections~4.2 and 4.3]{VoZeMa2014FFTH} for an odd approximation order \eqref{eq:N_odd} and in \cite{VoZeMa2014GarBounds} for a general one.}

\begin{definition}[Galerkin approximation --- Ga]
\label{def:Ga}
The primal and the dual homogenization matrices $\AeffN,\BeffN\in \xRdd$ of Ga satisfy
\begin{subequations}
\label{eq:Ga}
\begin{align}
\label{eq:GA_prim}
\AeffN \VE\cdot\VE &= \inf_{\Ve_\VN\in\cEN} \bilf{\VE+\Ve_\VN}{\VE+\Ve_\VN} = \bilf{\VE+\Ve_\VN\mac{\VE}}{\VE+\Ve_\VN\mac{\VE}}
\\
\BeffN\VJ\cdot\VJ 
&= \inf_{\Vj_\VN\in\cJN} \bilfi{\VJ+\Vj_\VN}{\VJ+\Vj_\VN} = \bilfi{\VJ+\Vj_\VN\mac{\VJ}}{\VJ+\Vj_\VN\mac{\VJ}}
\end{align}
\end{subequations}
for arbitrary quantities $\VE,\VJ\in\xRd$.
\end{definition}
\begin{remark}Similarly to Remark~\ref{eq:eqiv2pot} for a continuous setting,
the primal formulation \eqref{eq:GA_prim} coincides with the problem \eqref{eq:intro_Ga}, because the subspace $\cEN$ from \eqref{eq:Helmholtz_trig} admits an equivalent characterization
$\cEN = \{ \nabla f:f\in \cTN \}$.
\end{remark}
Evaluation of the integrals in \eqref{eq:homog_problem}, described in section~\ref{sec:eval_integrals}, is generally unfeasible in closed form. This weakness can be compensated for with the trapezoidal (rectangular) integration rule leading to the approximated bilinear forms $\bilfN{}{}, \bilfNi{}{} :\cTNd\times\cTNd\rightarrow\xR$ expressed as
\begin{subequations}
\label{eq:approx_bilin}
\begin{align}
\bilf{\Vu_\VN}{\Vv_\VN}\approx\bilfN{\Vu_\VN}{\Vv_\VN} &:= \sum_{\Vk\in\ZNd}\TA(\Vx_\VN^\Vk)\Vu_\VN(\Vx_\VN^\Vk)\Vv_\VN(\Vx_\VN^\Vk),
\\
\bilfi{\Vu_\VN}{\Vv_\VN}\approx\bilfNi{\Vu_\VN}{\Vv_\VN} &:= \sum_{\Vk\in\ZNd}\TA^{-1}(\Vx_\VN^\Vk)\Vu_\VN(\Vx_\VN^\Vk)\Vv_\VN(\Vx_\VN^\Vk).
\end{align}
\end{subequations}
Noting that the objects related to this numerical integration are consistently denoted with the tilde symbol, e.g. $\bilfN{}{}$, $\gVe_\VN\mac{\VE}$, $\tAeffN$, or $\tAeffN\bound$.
\begin{definition}[Galerkin approximation with numerical integration --- GaNi]
\label{def:GaNi}
Let material coefficients \eqref{eq:A} be additionally Riemann integrable or continuous $\TA\in\Cper{0}{\xRddspd}$. Then, the primal and the dual homogenized coefficients $\tAeffN, \tBeffN\in\xR^{d\times d}$ satisfy
\begin{subequations}
\label{eq:GaNi}
\begin{align}
\label{eq:GaNi_primal}
\tAeffN\VE\cdot\VE &= \inf_{\Ve_{\VN}\in \cEN } \bilfN{\VE+\Ve_{\VN}}{\VE+\Ve_{\VN}} = \bilfN{\VE+\gVe_{\VN}\mac{\VE}}{\VE+\gVe_{\VN}\mac{\VE}},
\\
\label{eq:GaNi_dual}
\tBeffN\VJ\cdot\VJ &= \inf_{\Vj_{\VN}\in \cJN} \bilfNi{\VJ+\Vj_{\VN}}{\VJ+\Vj_{\VN}} = \bilfNi{\VJ+\gVj_{\VN}\mac{\VJ}}{\VJ+\gVj_{\VN}\mac{\VJ}},
\end{align}
\end{subequations}
for arbitrary quantities $\VE,\VJ\in\xRd$.
\end{definition}
\begin{remark}[Duality in GaNi]
\label{rem:duality_gani}
Surprisingly, according to \cite[Proposition~34]{VoZeMa2014GarBounds}, the duality of homogenized matrices from the continuous formulation \eqref{eq:Aeff_duality} is inherited in the GaNi for the odd grids \eqref{eq:N_odd}, particularly
\begin{align*}
\tBeffN^{-1} = \tAeffN. 
\end{align*}
\end{remark}
In \cite[section~6]{VoZeMa2014GarBounds}, the minimizers of  
the GaNi scheme $\gVe_\VN\mac{\VE},\gVj_\VN\mac{\VJ}$
were used to a posteriori define the approximate homogenized coefficients $\tAeffN\bound, \tBeffN\bound\in\xRdd$ satisfying for all $\VE,\VJ\in\xRd$
\begin{subequations}
\label{eq:GaNi_apost}
\begin{align}
\tAeffN\bound \VE\cdot\VE &=
\bilf{\VE+\gVe_\VN\mac{\VE}}{\VE+\gVe_\VN\mac{\VE}},
\\
\tBeffN\bound\VJ\cdot\VJ &=
\bilfi{\VJ+\gVj_\VN\mac{\VJ}}{\VJ+\gVj_\VN\mac{\VJ}}.
\end{align}
\end{subequations}
\subsection{Structure of guaranteed bounds on homogenized properties}
\label{sec:bounds}
In \cite[section~6]{VoZeMa2014GarBounds}, we have shown that the matrices in \eqref{eq:GaNi_apost} are guaranteed bounds on homogenized coefficients, i.e. $$\bigl(\tBeffN\bound\bigr)^{-1} \preceq \Beff^{-1} = \Aeff \preceq  \tAeffN\bound.$$
where the L\"{o}wner partial order $\preceq$ is used on a space of symmetric positive definite matrices $\xRddspd$ according to~\cite[section~7.7]{horn2012matrix} as
\begin{align*}
\text{for }\T{L},\T{M}\in\xRddspd:\quad \T{L}\preceq\T{M}\quad\text{if}\quad \T{L} \Vu\cdot\Vu \leq \T{M} \Vu \cdot \Vu \quad\text{for all }\Vu\in\xRd.
\end{align*}
Here, this structure is extended for the homogenized coefficients defined with the Ga scheme \eqref{eq:Ga}.
\begin{proposition}[Structure of guaranteed bounds]
\label{lem:struct_bounds}
The homogenized coefficients occurring in equations \eqref{eq:homog_problem}, \eqref{eq:Ga}, and \eqref{eq:GaNi_apost} have the following structure
\begin{align}\label{eq:structure_homog_matrices}
\bigl(\tBeffN\bound\bigr)^{-1} \preceq \BeffN^{-1} \preceq \Beff^{-1} = \Aeff \preceq \AeffN \preceq \tAeffN\bound.
\end{align}
\end{proposition}
\begin{proof}
First of all, the proof of inequality $\Aeff\preceq\AeffN$ is based on the conformity of approximate space $\cEN\subset\cE$, recall~\eqref{eq:Helmholtz_trig}. Indeed, as
the minimization space of the homogenization problem \eqref{eq:homog_problem_primal} is reduced, the minimum has to remain or increase, i.e.
\begin{align*}
\Aeff\VE\cdot\VE &= \inf_{\Ve\in\cE} \bilf{\VE+\Ve}{\VE+\Ve}
\leq \inf_{\Ve_\VN\in\cEN} \bilf{\VE+\Ve_\VN}{\VE+\Ve_\VN} = \AeffN\VE\cdot\VE.
\end{align*}
Now, the following inequality $\AeffN\preceq\tAeffN$ is proven by substituting the minimizer $\Ve_\VN\mac{\VE}$ in 
the GA scheme \eqref{eq:GA_prim} with the approximate minimizers $\gVe_\VN\mac{\VE}$ of the GaNi scheme \eqref{eq:GaNi_primal}, i.e.
\begin{align*}
\AeffN\VE\cdot\VE &= \inf_{\Ve_\VN\in\cEN} \bilf{\VE+\Ve_\VN}{\VE+\Ve_\VN}
\leq \bilf{\VE+\gVe_\VN\mac{\VE}}{\VE+\gVe_\VN\mac{\VE}} = \tAeffN\bound\VE\cdot\VE.
\end{align*}
Since the primal inequalities $\Aeff \preceq \AeffN \preceq \tAeffN$ are in hand, the dual formulation reveals the upper bounds on the dual matrix
$\Beff \preceq \BeffN \preceq \tBeffN\bound$ according to the same arguments.
The proof then arises from the inverse inequality, according to \cite[Corollary~7.7.4.(a)]{horn2012matrix}, i.e.
\begin{align}
\label{eq:spd_inverse_ineq}
\text{for }\T{L},\T{M}&\in\xRddspd
&
 \T{L} \preceq \T{M}
 \quad &\Longleftrightarrow\quad
  \T{M}^{-1} \preceq \T{L}^{-1},
\end{align}
 and from the duality of the continuous homogenization problem \eqref{eq:Aeff_duality}, i.e.
$\Beff^{-1} = \Aeff$,
which follows by the standard duality arguments in \cite{ekeland1976convex,Suquet1982dual} or  \cite[Proposition~7 and Corollary~9]{VoZeMa2014GarBounds}.
\end{proof}

\section{Efficient numerical integration}
\label{sec:eval_integrals}

The computation of the guaranteed bounds on the homogenized coefficients \eqref{eq:structure_homog_matrices} consists of the evaluation of bilinear forms \eqref{eq:bilinear_forms} occurring in the formulation of the Ga scheme \eqref{eq:Ga} or in the a posteriori estimate with the GaNi minimizers \eqref{eq:GaNi_apost}, i.e. integrals of the type
\begin{align}\label{eq:integral2eval}
\scal{\TA \Vu_{\VN} }{ \Vv_{\VN} }
_{\Lper{2}{\xRd}}\quad\text{ for }\TA\in\Lper{\infty}{\xRdd}\text{ and }\Vu_\VN,\Vv_\VN\in\cTNd.
\end{align}

This integral evaluation based on double grid quadrature has already been analyzed in \cite[section~6]{VoZeMa2014GarBounds}, which is summarized in section~\ref{sec:calculation_bounds} for a matrix-inclusion composite \eqref{eq:inclcomp}.
Then, in section~\ref{sec:grid_composite}, the methodology is generalized to an effective evaluation of integral \eqref{eq:integral2eval} for grid-based composites, which are defined via high-resolution images assuming piece-wise constant or piece-wise bilinear material coefficients.

In section~\ref{sec:reduce_grid}, I show that the double grid quadrature can be reduced to the original grid, resulting in a reduction of memory requirements.
In section~\ref{sec:approx_bounds}, I show that the guaranteed bounds on homogenized coefficients are not confined by a closed-form evaluation of \eqref{eq:integral2eval} requiring the knowledge of Fourier coefficient $\FT{\TA}(\Vk)$;
however, the approximation of $\TA$ can be used, leading to upper-upper and lower-lower bounds.

\subsection{Notation}
\label{sec:notation_discr}
A multi-index notation is systematically employed in which $\xX^{\VM}$
represents $\xX^{M_1 \times \cdots \times M_\dime}$ for $\VM\in\set{N}^d$. Then the sets $\xXM$ and
$\xMM$ or their complex counterparts represent the space of vectors
and matrices denoted by bold serif font, e.g. $\MBu =
\bigl(u_{\alp}^\Vk\bigr)_{\alp}^{\Vk\in\ZMd} \in \xXM$ and $\MB{U} =
(U_{\alp\bet}^{\Vk\Vm})_{\alp,\bet}^{\Vk,\Vm\in\ZMd} \in \xMM$
with the index set $\ZMd$ defined in \eqref{eq:index_set}. 

Sub-vectors and
sub-matrices are designated by superscripts, e.g. $\MBu^\Vk =
(u_{\alp}^\Vk)_{\alp} \in \xRd$ or $\MB{U}^{\Vk\Vm} =
(U_{\alp\bet}^{\Vk\Vm})_{\alp,\bet} \in \xRdd$. The scalar products on $\xXM$ and $\xhXM$
are defined as
\begin{align*}
\scal{\MB{u}}{\MBv}_{\xXM} 
&= 
\frac{1}{\meas{\VM}}
\sum_{\Vk\in\ZMd} 
\MBu^{\Vk}\cdot\MBv^{\Vk},
&
\scal{\MB{u}}{\MBv}_{\xhXM} 
&= 
\sum_{\Vk\in\ZMd} 
\MBu^{\Vk}\cdot\conj{\MBv^{\Vk}},
\end{align*}
where $|\VM| = \prod_\alp M_\alp$ stands for the number of discretization points.
Moreover, the matrix-vector or matrix-matrix multiplications
follow from
\begin{align*}
(\MB{U}\MBv)^{\Vk} 
= 
\sum_{\Vm\in\ZNd}
\MB{U}^{\Vk\Vm}
\MBv^{\Vm}
\in \xRd
\quad
\text{or}
\quad
(\MB{U} \MB{V})^{\Vk\Vm} 
=
\sum_{\Vn\in\ZNd}
\MB{U}^{\Vk\Vn}\MB{V}^{\Vn\Vm}
\in \xRdd,
\end{align*}
for $\Vk,\Vm \in \ZNd$ and $\MB{V} \in \xMM$.

The discretization operator that stores the values of a function at grid points \eqref{eq:grid_points} is defined by
\begin{align}
\label{eq:IN}
\IM{\VM}&:\Cper{0}{\xRd}\rightarrow\xXM,
&
\IM{\VM}[\Vu] &= \bigl( u_{\alp}(\Vx_\VM^{\Vm}) \bigr)_{\alp=1,\dotsc,d}^{\Vm\in\ZMd}.
\end{align}
When the polynomial of order $\VN$ is expressed on a grid with the higher number of points $\VM$,
\begin{align*}
\MBu_\VN = \IM{\VM}[\Vu_\VN] \in \xR^{d\times\VM}
\quad\text{for }
\Vu_\VN\in\cTNd
\text{ and }
\VM,\VN\in\xRd
\text{ such that }
M_\alp>N_\alp,
\end{align*}
the discrete representation of the polynomial $\MBu_\VN$ is kept with subscript $\VN$ to emphasize the polynomial order rather than the vector size.
The actual dimension of $\MBu_\VN$ is understood implicitly from the context, so that terms such as $\scal{\MBA_\VM\MBu_\VN}{\MBu_\VN}_{\xR^{d\times\VM}}$ with $\MBA_\VM\in\bigl[\xR^{d\times\VM}\bigr]^2$ remain well-defined.

Using the discretization operator \eqref{eq:IN}, the relation \eqref{eq:trig_conn_raw} between trigonometric polynomials expressed in Fourier and space domain can be represented in compact form as
\begin{align}\label{eq:trig_conn}
\Vu_{\VN} &=  \sum_{\Vk\in\ZNd}\MBhu_\VN^{\Vk}\varphi_{\Vk} = \sum_{\Vm\in\ZMd}{\MBu_{\VN}^{\Vm}\varphi_{\VM,\Vm}}
\quad\text{with }
\MBu_\VN = \IM{\VM}[\Vu_\VN]
\text{ and }
\MBhu_\VN = \MBF_{\VN} \MBu_\VN\in\xCdN,
\end{align}
where matrices of vector-valued DFT and inverse DFT are defined as
\begin{align}\label{eq:DFT}
\MBF_{\VN} &= \frac{1}{\pVN} \bigl( \del_{\alp\beta} \omega_{\VN}^{-\Vm\Vk} \bigr)_{\alp,\beta}^{\Vm,\Vk\in\ZNd} \in\xhMN,
&
\MBFi_{\VN} &= \bigl( \del_{\alp\beta} \omega_{\VN}^{\Vm\Vk} \bigr)_{\alp,\beta}^{\Vm,\Vk\in\ZNd} \in\xhMN.
\end{align}
\begin{lemma}\label{lem:IN_iso}
For odd $\VN$ according to \eqref{eq:N_odd} and $\VM\in\xNd$ such that $M_\alp\geq N_\alp$,
discretization operator $\IM{\VM}$ is an isometric isomorphism (one-to-one map that preserves distances) between trigonometric polynomials $\cTNd$ defined in \eqref{eq:trig_space} and vector space $\IM{\VM}[\cTNd]\subseteq\xRdM$, with equality $\IM{\VN}[\cTNd]=\xRdN$ for $\VN=\VM$. 
\end{lemma}

\subsection{Methodology}
\label{sec:calculation_bounds}
Here, the basic concepts regarding numerical integration of \eqref{eq:integral2eval} are summarized according to \cite{VoZeMa2014GarBounds}.
The evaluation of bilinear forms in the conventional FEM leads to sparse matrices, which also is the case for the GaNi \eqref{eq:GaNi}.
However, more complicated scenarios arise with the Ga~\eqref{eq:Ga}. A direct integration  on the original grid leads to a fully populated matrix, see Lemma~\ref{lem:FD_matrix_full}, while a double-grid quadrature produces a sparse matrix with the same block diagonal structure as the GaNi; see Lemma~\ref{lem:integ_eval_VM} and compare with Remark~\ref{rem:FD_GaNi}.
\begin{remark}[Rectangular integration rule]\label{rem:FD_GaNi}Using the notation about discrete spaces from section~\ref{sec:notation_discr}, the numerical integration in GaNi, recall \eqref{eq:approx_bilin}, leads to the following expression
\begin{align*}
\bilfN{\Vu_\VN}{\Vv_\VN} = \scal{\MBtA_\VN \MBu_\VN}{\MBv_\VN}_{\xRdN}
\quad
\text{for }
\Vu_\VN,\Vv_\VN\in\cTNd
\text{ and } 
\MBu_\VN = \IM{\VN}[\Vu_\VN], \MBv_\VN = \IM{\VN}[\Vv_\VN]
\end{align*}
with $\MBtA_\VN^{\Vk\Vl} = \del_{\Vk\Vl}\TA(\Vx_\VN^\Vk)$ for $\Vk,\Vl\in\ZNd$.
\end{remark}
\begin{lemma}[Fully populated expression]
\label{lem:FD_matrix_full}
According to \cite[Lemma~37]{VoZeMa2014GarBounds}, the integral \eqref{eq:integral2eval} equals
\begin{align*}
\scal{\TA\Vu_\VN}{\Vv_\VN}_{\Lper{2}{\xRd}} &= 
\scal{\MBhA^{\full}_\VN
\MBhu_\VN}{\MBhv_\VN}_{\xhXN} = \scal{\MBA^{\full}_\VN\MBu_\VN}{\MBv_\VN}_{\xXN}
\end{align*}
where vectors $\MBu_\VN,\MBv_\VN\in\xXN$ and $\MBhu_\VN,\MBhv_\VN\in\xhXN$ are defined via
\begin{align*}
\MBu_\VN &= \IM{\VN}[\Vu_\VN],
&
\MBv_\VN &= \IM{\VN}[\Vv_\VN],
&
\MBhu_\VN &= \MBF_\VN\MBu_\VN,
&
\MBhv_\VN &= \MBF_\VN\MBv_\VN,
\end{align*}
and matrices $\MBA^{\full}_\VN\in\xMN$ and $\MBhA^{\full}_\VN\in\xhMN$ as
\begin{align*}
\bigl( \MBhA^{\full}_\VN \bigr)^{\Vl\Vk}
&= \int_{\puc} \TA(\Vx) \varphi_\Vk(\Vx)  \varphi_{-\Vl}(\Vx) \D{\Vx}
\quad\text{for }\Vl,\Vk\in\ZNd,
&
\MBA^{\full}_\VN &= \MBF_\VN \MBhA^{\full}_\VN \MBFi_\VN.
\end{align*}
\end{lemma}
\begin{lemma}[Sparse expression on a double grid, Lemma~39 in \cite{VoZeMa2014GarBounds}]
\label{lem:integ_eval_VM}
For an odd order $\VN$ according to \eqref{eq:N_odd} and $\VM$ such that $M_\alp\geq 2N_\alp -1$, 
the integral \eqref{eq:integral2eval} equals
\begin{align*}
\scal{\TA\Vu_\VN}{\Vv_\VN}_{\Lper{2}{\xRd}} &=
\scal{\MBA_\VM\MBu_\VN}{\MBv_\VN}_{\xRdM}
\end{align*}
where $\MBu_\VN = \IM{\VM}[\Vu_\VN]$, $\MBv_\VN=\IM{\VM}[\Vv_\VN]\in\xXM$, and
$\MBA_\VM\in\xMM$ with components
\begin{align}\label{eq:MBA_GA}
(\MBA_{\VM})_{\alp\beta}^{\Vk\Vl} &=
\del_{\Vk\Vl} \sum_{\Vn\in\ZtNrd}
\omega_{\VM}^{\Vk\Vn}
\FT{A}_{\alp\beta}(\Vn)
\quad
\text{for }\Vk,\Vl\in\ZMd\text{ and }\alp,\beta=1,\dotsc,d.
\end{align}
\end{lemma}
Because of the requirements in section~\ref{sec:reduce_grid}, this lemma generalizes \cite[Lemma~39]{VoZeMa2014GarBounds} that confines $\VM=2\VN-\V{1}$; thus, a proof is presented here.
\begin{proof}
Because the product of two trigonometric polynomials $\Vu_\VN \Vv_{\VN}=\bigl( u_{\VN,\alp} v_{\VN,\alp}\bigr)_{\alp}\in\cTtNrd$ has an order bounded by $\tVNr$, it can be expressed, in accordance with \eqref{eq:trig_conn_raw}, on any grid $\VM$ such that $M_\alp\geq 2N_\alp-1$, i.e.
\begin{align*}
u_{\VN,\beta}v_{\VN,\alp} &= \sum_{\Vm\in\ZMd} \FT{u_{\VN,\beta}v_{\VN,\alp}}(\Vm) \varphi_\Vm = \sum_{\Vm\in\ZMd} u_{\VN,\beta}(\Vx_\VM^\Vm)
v_{\VN,\alp}(\Vx_\VM^\Vm) \varphi_{\VM,\Vm} 
\end{align*}
where $\FT{u_{\VN,\beta}v_{\VN,\alp}}(\Vm) = 0$ for
$\Vm\in\ZMd \setminus \ZtNrd$.

Substitution into \eqref{eq:integral2eval} and direct computation reveals
\begin{align*}
\scal{\TA\Vu_\VN}{\Vv_\VN}_{L^2_{\per}} &= \sum_{\alp,\beta} \sum_{\Vm\in\ZtNrd} \FT{A}_{\alp\beta}(-\Vm) \FT{\Vu_{\VN,\beta}\Vv_{\VN,\alp}}(\Vm)
\\
&= 
\sum_{\alp,\beta} \sum_{\Vm\in\ZtNrd} \FT{A}_{\alp\beta}(-\Vm)
\sum_{\Vk\in\ZMd} \frac{\omega_\VM^{-\Vm\Vk}}{\pVM} u_{\VN,\beta}(\Vx_\VM^\Vk) v_{\VN,\alp}(\Vx_\VM^\Vk).
\end{align*}
The statement of the lemma follows by substituting $\Vn$ with $-\Vm$.
\end{proof}
\begin{remark}[Material coefficients leading to guaranteed bounds]
\label{rem:anal_Fourier}
The closed-form evaluation of integral \eqref{eq:integral2eval} leading to the fully discrete matrix \eqref{eq:MBA_GA} rests upon recognition of the Fourier series expansion for material coefficients $\TA$. The space of functions having analytical expression of Fourier coefficients constitutes a linear space thanks to the linearity of the integration. Three suitable examples, which are also used for the numerical examples in section~\ref{sec:numerical_examples}, are introduced here.

Let $\Vh\in\xR^d$ and $r\in\xR$ be parameters such that $0<h_\alp\leq 
1$ and $r\leq 1$, then the periodic 
functions
$\rect_{\Vh},\tri_{\Vh},
\circl_r\in L^{\infty}_{\per}(\puc;\xR)$ are defined on
$\puc$ along with their Fourier coefficients as
\begin{subequations}
\label{eq:char_func_Four}
\begin{align}
\label{eq:def_cube}
\rect_{\Vh}(\Vx) &=
\begin{cases}
1
&
\text{if }
|x_\alp|<\frac{h_\alp}{2}\text{ for all }\alp
\\
0
&
\text{otherwise}
\end{cases},
&
\FT{\rect}_{\Vh}(\Vk) &
= \prod_{\alp} h_\alp \sinc\left(
h_\alpha k_{\alp} \right),
\\
\label{eq:def_tri}
\tri_{\Vh}(\Vx) &= \prod_{\alp}\max\{1-|\frac{x_\alp}{h_\alp}|,0\},
&
\FT{\tri}_{\Vh}(\Vk) &
= \prod_{\alp} h_{\alp} \sinc^2
\left(h_\alpha k_{\alp} \right),
\\
\circl_r(\Vx) &= 
\begin{cases}
1
&
\text{for }\|\Vx\|_2 < r\\
0
&
\text{otherwise}
\end{cases},
&
\FT{\circl}_{r}(\Vk) &
= 
\begin{cases}
\pi r^2
&
\text{for }\Vk=\Vo
\\
r^2\frac{B_1(2\pi r\|\Vk\|_2)}{r\|\Vk\|_2}
&
\text{otherwise}
\end{cases},
\end{align}
\end{subequations}
where 
$\sinc (x) = 
\begin{cases}
1
&
\text{for }x=0
\\
\frac{\sin(\pi x)}{\pi x}
&
\text{for }x\neq 0
\end{cases}
$ and $B_1$ is the Bessel function of the first kind.
\end{remark}
In \cite[Lemma~38]{VoZeMa2014GarBounds}, 
the incorporation of characteristic functions \eqref{eq:char_func_Four} has been elaborated in detail for the inclusion-matrix composites, characterized  by coefficients in the form
\begin{align}\label{eq:inclcomp}
\TA(\Vx) &= \TA\incl{0}+\sum_{i=1}^n f\incl{i} (\Vx-\Vx\incl{i}) \TA\incl{i}
\end{align}
where matrices $\TA\incl{i}\in\xRdd$ for $i=0,\dotsc,n$ represent the coefficients of the matrix phase and inclusions, functions $f\incl{i}\in L^{\infty}_\per(\puc)$
quantify the distribution of coefficients within inclusions, centered at $\Vx\incl{i}$, along with their topology. 
Their discrete coefficients \eqref{eq:MBA_GA} read as
\begin{align}\label{eq:inclcomp_formula}
\MBA_{\VM}^{\Vk\Vl} 
&= \del_{\Vk\Vl} \left[ \TA\incl{0}  + 
\sum_{i=1}^n \TA\incl{i} 
\left( \sum_{\Vm\in\ZtNrd} \omega_{\VM}^{\Vk\Vm}
\varphi_{-\Vm}(\Vx\incl{i})
\FT{f}\incl{i}(\Vm) \right)
\right]
\quad\text{for }\Vk,\Vl\in\ZMd.
\end{align}
\subsection{Material coefficients defined on grids}
\label{sec:grid_composite}
As an alternative to the inclusion-matrix composite \eqref{eq:inclcomp}, the grid-based composite is expressed as a linear combination of characteristic functions $\psi$ concentrated on grid points $\Vx_\VP^\Vm$ for $\Vm\in\ZPd$ and $\VP\in\xNd$, i.e.
\begin{align}\label{eq:gridcomp}
\TA(\Vx) &=
\sum_{\Vp\in\ZPd}\psi(\Vx-\Vx^{\Vp}_{\VP})\MBC_\VP^{\Vp}
\quad\text{for }\VP\in\xNd, \Vx\in\puc,\text{ and }\MBC_\VP\in \xR^{d\times d\times\VP}.
\end{align}
The material resolution is systematically expressed by the symbol $\VP$, which corresponds to the number of grid points and is independent of the order $\VN\in\xNd$ of trigonometric polynomials $\cTNd$. 
\begin{remark}
The material coefficients \eqref{eq:gridcomp} are represented with a pixel- or voxel-based image with a resolution $\VP\in\xNd$. If the basis function $\psi$ is taken as \eqref{eq:def_cube} or \eqref{eq:def_tri}, the coefficients are piece-wise constant or bilinear, respectively.
\end{remark}
\begin{lemma}[Grid-based composites]
\label{lem:grid_compos}
The matrix \eqref{eq:MBA_GA} for material coefficients \eqref{eq:gridcomp} is given by
\begin{align}
\label{eq:gridcomp_formula}
\bigl( \MBA_{\VM} \bigr)^{\Vk\Vl}_{\alp\beta}
&= \del_{\Vk\Vl}\pVP
\sum_{\Vm\in\ZtNrd} \omega_{\VM}^{\Vk\Vm} \FT{\psi} (\Vm) 
\left( \sum_{\Vp\in \ZPd} \frac{\omega_{\VP}^{-\Vm\Vp}}{\pVP} 
 \M{C}_{\VP,\alp\beta}^{\Vp} \right)
\end{align}
for 
 $\Vk,\Vl\in\ZMd$ and $\alp,\beta=1,\dotsc,d$,
where $\FT{\psi}(\Vm)$ for $\Vm\in\ZtNrd$ are the Fourier coefficients of $\psi$.
\end{lemma}
\begin{proof}
Using an affine substitution, the following formula for $\Vm,\Vp\in\xZ^d$ is deduced
\begin{align*}
\int_{\puc} {\psi(\Vx-\Vx^{\Vp}_{\VP}) \varphi_{-\Vm}(\Vx)}{\D{\Vx}} 
= \int_{\puc} {\psi(\Vx)\varphi_{-\Vm}(\Vx+\Vx^{\Vp}_{\VP})} {\D{\Vx}} 
=  \omega_{\VP}^{-\Vm\Vp} \FT{\psi}(\Vm).
\end{align*}
This, \eqref{eq:MBA_GA}, and \eqref{eq:gridcomp} enable computation
\begin{align*}
\bigl( \MBA_{\VM} \bigr)_{\alp\beta}^{\Vk\Vl} &= \del_{\Vk\Vl} \sum_{\Vm\in\ZtNrd}
\omega_{\VM}^{\Vk\Vm}
\int_{\puc}
A_{\alp\beta}(\Vx) \varphi_{-\Vm}(\Vx)\D{\Vx}
\\
&= \del_{\Vk\Vl} \sum_{\Vm\in\ZtNrd}
\omega_{\VM}^{\Vk\Vm}
\sum_{\Vp\in\ZPd} \M{C}_{\VP,\alp\beta}^{\Vp}
\int_{\puc}
\psi(\Vx-\Vx^{\Vp}_{\V{P}}) \varphi_{-\Vm}(\Vx)\D{\Vx}
\\
&= \del_{\Vk\Vl} \sum_{\Vm\in\ZtNrd}
\omega_{\VM}^{\Vk\Vm}
\sum_{\Vp\in\ZPd} \M{C}_{\VP,\alp\beta}^{\Vp} \omega_{\VP}^{-\Vm\Vp} \FT{\psi}(\Vm). 
\end{align*}
\end{proof}

\subsection{Reduction from a double to the original grid}
\label{sec:reduce_grid}
In Lemma~\ref{lem:integ_eval_VM}, the integral \eqref{eq:integral2eval} occurring in the Ga scheme \eqref{eq:Ga} has been evaluated on a double grid of size $\VM$ such that $M_\alp\geq 2N_\alp-1$.
Here in Lemma~\ref{lem:reduce_grid}, it is reformulated using DFT shifts to the original size $\VN$ occurring also in the GaNi scheme, cf.~Remark~\ref{rem:FD_GaNi}.
This procedure reduces memory requirements which is discussed in Remark~\ref{rem:mem_req}.
\begin{lemma}[Reduction to the original grid]
\label{lem:reduce_grid}
Let $\xSd=\prod_\alp\{0,1\}$ and $\MBA_\tVN$ be a matrix from Lemma~\ref{lem:integ_eval_VM}. Then
the fully populated matrices in Lemma~\ref{lem:FD_matrix_full}
can be reformulated to
\begin{subequations}
\label{eq:Afull_decomp}
\begin{align}
\label{eq:Afull_decomp_space}
\MBA^\full_\VN &= 2^{-d}\sum_{\Vs\in\xSd} \MBFi_\VN \MBS_\VN^*(\Vs) 
\MBF_\VN 
\MBA_{\VN}
(\Vs)\MBFi_\VN\MBS_\VN(\Vs) \MBF_\VN,
\\
\label{eq:Afull_decomp_Fourier}
\MBhA^\full_\VN  &= 2^{-d} \sum_{\Vs\in\xSd} \MBS_\VN^*(\Vs) 
\MBF_\VN 
\MBA_{\VN}
(\Vs)\MBFi_\VN\MBS_\VN(\Vs),
\end{align}
\end{subequations}
where the matrices
$\MBS_\VN(\Vs),\MBA_{\VN}(\Vs)\in\xhMN$ are defined for $\Vs\in\xSd$ as
\begin{align*}
\MBS_\VN(\Vs) &=
( \del_{\alp\beta}\del_{\Vm\Vn} \omega_\tVN^{-\Vs\Vm})^{\Vm,\Vn\in\ZNd}_{\alp,\beta},
&
\MBA_{\VN} (\Vs) &= (\del_{\Vk\Vr}\MBA^{(2\Vr-\Vs)(2\Vr-\Vs)}_{\tVN,\alp\beta})_{\alp,\beta}^{\Vk,\Vr\in\ZNd}.
\end{align*}
\end{lemma}
\begin{proof}
According to Lemma~\ref{lem:FD_matrix_full} and \ref{lem:integ_eval_VM}, integral \eqref{eq:integral2eval} can be expressed in three ways
\begin{align}\label{eq:bilinear_tVN_VN}
 \scal{\TA\Vu_\VN}{\Vv_\VN}_{\Ltp} &= \scal{\MBA_\VN^\full\MBu_\VN}{\MBv_\VN}_{\xRdN} = \scal{\MBhA_\VN^\full\FT{\MBu}_\VN}{\FT{\MBv}_\VN}_{\xhXN}
 = \scal{\MBA_\tVN\bMBu_\VN}{\bMBv_\VN}_{\xRdtN}
\end{align}
where 
\begin{align*}
\Vu_\VN &= \IMi{\VN}[\MBu_\VN] = \IMi{\VN}[\MBFi_\VN\FT{\MBu}_\VN] = \IMi{\tVN}[\bMBu_\VN],
&
\Vv_\VN &= \IMi{\VN}[\MBv_\VN] = \IMi{\VN}[\MBFi_\VN\FT{\MBv}_\VN] = \IMi{\tVN}[\bMBv_\VN].
\end{align*}

Now, the matrix $\MBA_\tVN$ will be decomposed to meet \eqref{eq:Afull_decomp}.
In order to reduce the double grid sizing $\tVN$ to original grid sizing $\VN$, index $\Vk\in\ZtNd$ is uniquely split into 
\begin{align}\label{eq:index_split}
\Vk = 2\Vr - \Vs
\quad\text{with }\Vr\in\ZNd\text{ and }\Vs\in\xS. 
\end{align}
Using the connection of the two representations of trigonometric polynomials via DFT stated in~\eqref{eq:trig_conn}, it holds for $\Vk\in\ZtNd$ according to \eqref{eq:index_split} that
\begin{align*}
\bMBu^{\Vk}_\VN &= \Vu_\VN(\Vx^\Vk_\tVN) 
= \sum_{\Vm\in\ZNd} 
\omega_\tVN^{\Vk\Vm}\FT{\Vu}_\VN(\Vm)
= \sum_{\Vm\in\ZNd}
\omega_\VN^{\Vr\Vm}\omega_\tVN^{-\Vs\Vm} \FT{\Vu}_\VN(\Vm)
= \bigl[ \MBFi_\VN \MBS_\VN(\Vs) \FT{\MBu}_\VN \bigr]^{\Vr},
\end{align*}
where summation occurs only over $\ZNd$ instead of $\ZtNd$ since the trigonometric polynomial $\Vu_\VN$ belongs to $\cTNd$.
Using the notation $\MBA_\tVN^{\Vk}:=\MBA_\tVN^{\Vk\Vk}\in\xRdd$ for $\Vk\in\ZtNd$, the bilinear form \eqref{eq:bilinear_tVN_VN} can
be reformulated into
\begin{align*}
\scal{\MBA_\tVN\bMBu_\VN}{\bMBv_\VN}_{\xRdtN}
&= \frac{1}{\ptVN} \sum_{\Vk\in\ZtNd}\scal{\MBA_\tVN^{\Vk}\bMBu_\VN^\Vk}{\bMBv_\VN^\Vk}_{\xRd}
\\
&= \frac{1}{\ptVN} \sum_{\alp,\beta} \sum_{\Vs\in\xSd} \sum_{\Vr\in\ZNd}
\MA_{\tVN,\alp\beta}^{2\Vr-\Vs} \bMu_{\VN,\beta}^{2\Vr-\Vs}
\bMv_{\VN,\alp}^{2\Vr-\Vs}
\\
&= \frac{1}{2^d\pVN}\sum_{\alp,\beta} \sum_{\Vs\in\xSd} \sum_{\Vr,\Vm,\Vn\in\ZNd}
\MA_{\tVN,\alp\beta}^{2\Vr-\Vs} \left(
\omega_\VN^{\Vr\Vm}\omega_\tVN^{-\Vs\Vm} \FT{\Mu}_{\VN,\beta}^\Vm\right) 
\left( \omega_\VN^{\Vr\Vn}\omega_\tVN^{-\Vs\Vn}
\FT{\Mv}_{\VN,\alp}^{\Vn}\right).
\end{align*}
In the latter brackets, the sign of index $\Vn$ can be changed thanks to the symmetry of $\ZNd$ for odd grids \eqref{eq:N_odd}, i.e.
$\Vn\in\ZNd \Rightarrow -\Vn\in\ZNd$,
which allow for expressing summations as a scalar product on $\xhXN$, so
\begin{align*}
\scal{\MBA_\tVN\bMBu_\VN}{\bMBv_\VN}_{\xRdtN}
&= \frac{1}{2^d\pVN}\sum_{\alp,\beta} \sum_{\Vs\in\xSd} \sum_{\Vr,\Vm,\Vn\in\ZNd}
\MA_{\tVN,\alp\beta}^{2\Vr-\Vs} \left(
\omega_\VN^{\Vr\Vm}\omega_\tVN^{-\Vs\Vm} \FT{\Mu}_{\VN,\beta}^\Vm\right) 
\left( \omega_\VN^{-\Vr\Vn} \omega_\tVN^{\Vs\Vn}
\FT{\Mv}_{\VN,\alp}^{-\Vn}\right)
\\
&= \frac{1}{2^d}\sum_{\alp,\beta} \sum_{\Vs\in\xSd} \sum_{\Vr,\Vm,\Vn\in\ZNd} \ol{\omega_\tVN^{-\Vs\Vn}} \frac{ \omega_\VN^{-\Vr\Vn}}{\pVN}
\MA_{\tVN,\alp\beta}^{2\Vr-\Vs} \left(
\omega_\VN^{\Vr\Vm}\omega_\tVN^{-\Vs\Vm} \FT{\Mu}_{\VN,\beta}^\Vm\right) 
\left(
\ol{\FT{\Mv}_{\VN,\alp}^{\Vn}}\right)
\\
&= \frac{1}{2^d} \sum_{\Vs\in\xSd} \scal{\MBS_\VN(\Vs) \MBF_\VN \MBA_{\VN}
(\Vs)\MBFi_\VN\MBS_\VN(\Vs) \FT{\MBu}_\VN
}{\FT{\MBv}_\VN}_{\xhXN}
\end{align*}
where the conjugate symmetry of $\FT{\MBv}_\VN^{\Vn}$ and $\omega_\tVN^{\Vs\Vn}$ has been used. Using substitution
$\FT{\MBu}_\VN = \MBF_\VN\MBu_\VN$ and $\FT{\MBv}_\VN = \MBF_\VN\MBv_\VN$,
it can be reformulated as a scalar product on $\xXN$
\begin{align*}
\scal{\MBA_\tVN\bMBu_\VN}{\bMBv_\VN}_{\xRdtN}
&= \frac{1}{2^d} \sum_{\Vs\in\xSd} \scal{\MBFi_\VN \MBS_\VN(\Vs) \MBF_\VN \MBA_{\VN}
(\Vs)\MBFi_\VN\MBS_\VN(\Vs)\MBF_\VN\MBu_\VN
}{\MBv_\VN}_{\xXN}.
\end{align*}
\end{proof}
\subsection{Approximation of guaranteed bounds}
\label{sec:approx_bounds}
Throughout section~\ref{sec:eval_integrals}, the author presents the methodology for evaluating guaranteed bounds on homogenized properties relying on determination of the matrix \eqref{eq:MBA_GA}.
For closed-form evaluation, an analytical expression of Fourier coefficients $\FT{A}_{\alp\beta}(\Vm)$ for $\Vm\in\ZtNrd$ is required, cf.~Lemma~\ref{lem:integ_eval_VM}.

For general material coefficients, an approximate evaluation of the integral \eqref{eq:integral2eval} can violate the structure of guaranteed bounds \eqref{eq:structure_homog_matrices}; this can be resolved by appropriate adjustment of the material coefficients presented in the following Lemma~\ref{lem:bounds_bounds} with a particular example in Remark~\ref{rem:adjusting_A}. This approach is inspired by \cite{Dvorak1993master,Haslinger1995optimum}, which incorporated outer approximation of inclusion topology in the FEM framework.

\begin{lemma}[Upper-upper and lower-lower guaranteed bounds]
\label{lem:bounds_bounds}
Let $\overline{\TA},\underline{\TA}\in
L^{\infty}_{\per}(\puc;\xR^{d\times d})$ be upper and lower approximations of material coefficients \eqref{eq:A} satisfying
\begin{align}\label{eq:L_approx}
\underline{\TA}(\Vx) \preceq \TA(\Vx) \preceq
\overline{\TA}(\Vx)
\quad\text{for almost all }\Vx\in\puc,
\end{align}
and let $\obilf{}{}, \ubilfi{}{}:\Lper{2}{\xRd}\times\Lper{2}{\xRd}\rightarrow\xR$ be corresponding bilinear forms
\begin{align*}
\obilf{\Vu}{\Vv} &:=
 \scal{\ol{\TA} \Vu }{ \Vv }_{\Lper{2}{\xRd}},
&
\ubilfi{\Vu}{\Vv} &:= \scal{\ul{\TA}^{-1} \Vu}{
\Vv }_{\Lper{2}{\xRd}}.
\end{align*}
Then matrices $\ol{\TA}_{\eff}, \ul{\TA}_{\eff}, \ol{\TA}_{\eff,\VN}, \ul{\TA}_{\eff,\VN}\in\xRdd$  defined for arbitrary quantities $\VE, \VJ\in\xRd$
\begin{align*}
\scal{\ol{\TA}_{\eff} \VE}{\VE}_{\xRd} &=
 \inf_{\Ve\in\cE}\obilf{\VE + \Ve}{ \VE + \Ve },
&
\scal{\ol{\mB}_{\eff} \VJ}{\VJ}_{\xRd} &=
 \inf_{\Vj\in\cJ}\ubilfi{\VJ + \Vj}{ \VJ + \Vj },
\\
\scal{\ol{\TA}_{\eff,\VN} \VE}{\VE}_{\xRd} &=
 \inf_{\Ve\in\cE_\VN}\obilf{\VE + \Ve}{ \VE + \Ve },
&
\scal{\ol{\mB}_{\eff,\VN} \VJ}{\VJ}_{\xRd} &=
 \inf_{\Vj\in\cJ_\VN}\ubilfi{\VJ + \Vj}{ \VJ + \Vj },
 \\
\scal{\ol{\widetilde{\TA}}_{\eff,\VN} \VE}{\VE}_{\xRd} &=
 \obilf{\VE + \gVe_\VN\mac{\VE}}{ \VE + \gVe_\VN\mac{\VE}},
&
\scal{\ol{\widetilde{\mB}}_{\eff,\VN} \VJ}{\VJ}_{\xRd} &=
 \ubilfi{\VJ + \gVj_\VN\mac{\VJ}}{ \VJ + \gVj_\VN\mac{\VJ} }
\end{align*}
with minimizers $\gVe_\VN\mac{\VE},\gVj_\VN\mac{\VJ}$ of the GaNi scheme \eqref{eq:GaNi}
comply with the following structure of guaranteed bounds, i.e.
\begin{align*}
\begin{array}{ccccccccccccccc}
\ol{\widetilde{\mB}}_{\eff,\VN}^{-1} &\preceq& \ol{\mB}_{\eff,\VN}^{-1} & \preceq & \ol{\TB}_{\eff}^{-1} & \preceq & \Beff^{-1} & = & \Aeff & \preceq & \ol{\TA}_{\eff} & \preceq & \ol{\TA}_{\eff,\VN}  & \preceq & \ol{\widetilde{\TA}}_{\eff,\VN}
\\
 &  & \ol{\mB}_{\eff,\VN}^{-1} & \preceq & \BeffN^{-1}  & \preceq & \Beff^{-1} & = & \Aeff & \preceq & \AeffN & \preceq & \ol{\TA}_{\eff,\VN}  &  & 
\end{array}.
\end{align*}
\end{lemma}
\begin{proof}
It is possible to prove only the inequalities coming from the primal formulations since the dual part follows from the inverse inequality \eqref{eq:spd_inverse_ineq}. The inequalities $\Aeff \preceq
\AeffN$ and $\ol{\TA}_{\eff} \preceq \ol{\TA}_{\eff,\VN}\preceq \ol{\widetilde{\TA}}_{\eff,\VN}$ have already been proven in Proposition~\ref{lem:struct_bounds} for material coefficients $\TA$ and  $\ol{\TA}$, respectively.
\\
In order to prove the rest, the following inequality is deduced for arbitrary $\Vv\in\mathcal{X}\subseteq\cE$
\begin{align*}
\inf_{\Ve\in\mathcal{X}\subseteq\cE} \bilf{\VE + \Ve}{ \VE + \Ve} \leq \bilf{\VE + \Vv}{ \VE + \Vv} \leq  \obilf{\VE + \Vv}{ \VE + \Vv},
\end{align*}
where \eqref{eq:L_approx} and the monotonicity of the Lebesgue integration are used for the latter inequality. Since the first term is independent of $\Vv$, it is possible to add an infimum, i.e.
$\inf_{\Ve\in\mathcal{X}\subseteq\cE} \bilf{\VE + \Ve}{ \VE + \Ve} \leq \inf_{\Ve\in\mathcal{X}\subseteq\cE} \obilf{\VE + \Ve}{ \VE + \Ve}$.
The proof of $\Aeff\preceq \ol{\TA}_{\eff}$ and $\AeffN\preceq \ol{\TA}_{\eff,\VN}$ now follows for the choices $\mathcal{X} = \cE$ and $\mathcal{X}=\cEN$, respectively.
\end{proof}
\begin{remark}[Choice of $\ol{\TA}$ and $\ul{\TA}$] \label{rem:adjusting_A} To comply with requirement \eqref{eq:L_approx} in the previous lemma, 
a possible choice of material coefficients consists of local approximations with piece-wise constant functions in a grid-based composite \eqref{eq:gridcomp}. This material is then characterized with a pixel- or voxel-based image defined via the following formula
\begin{align*}
\ol{\TA} &= \sum_{\Vp\in\ZPd} \rect_\Vh (\Vx - \Vx_\VP^{\Vp}) \MBC_\VP^{\Vp}
\quad\text{with }
\MBC_\VP^{\Vp} = \norm{\TA(\cdot - \Vx_\VP^\Vp)}{L_\per^{\infty}(\Omega_\Vh;\xRdd)}\cdot\mI,
\end{align*}
where vector $\VP\in\xNd$ denotes an image resolution and where region $\Omega_\Vh = \Pi_\alp \bigl( -\frac{h_\alp}{2},\frac{h_\alp}{2} \bigr)$ for $h_\alp=\frac{1}{P_\alp}$ represents a pixel or voxel placed at the origin with characteristic function $\rect_\Vh$ defined in \eqref{eq:def_cube}. Factor $\norm{\TA(\cdot - \Vx_\VP^\Vk)}{L_\per^{\infty}(\Omega_\Vh;\xRdd)}$ then indicates the largest eigenvalue of material coefficients $\TA(\Vx)\in\xRdd$ over a pixel or voxel $(\Vx_\VP^\Vp+\Omega_\Vh)$ located at the corresponding grid point $\Vx_\VP^\Vp$.
\end{remark}
\begin{remark}The previous approximation of material coefficients, according to Lemma~\ref{lem:bounds_bounds}, leads to guaranteed bounds. The following approximation with piece-wise bilinear functions
\begin{align*}
 \TA(\Vx) \approx \sum_{\Vp\in\ZPd} \tri_\Vh (\Vx - \Vx_\VP^{\Vp})  \TA(\Vx_\VP^\Vp)
 \quad\text{for }h_\alp = \frac{1}{P_\alp}
\end{align*}
enables the closed-form computation of bilinear forms; however, the guaranteed bounds are only approximated.
\end{remark}

\section{Linear systems and computational aspects}
\label{sec:lin_sys}
The focus of this section is on the resolution of minimizers defined by the Galerkin approximations in section~\ref{sec:Galerkin_approx}. Using the results about numerical integration in section \ref{sec:eval_integrals}, the linear systems are described in section~\ref{sec:linear_system} with the help of discretization spaces of trigonometric polynomials, introduced in section~\ref{sec:discrete_sp}. Then, computational aspects are discussed in section~\ref{sec:computational_aspects}.
\subsection{Discretization spaces to trigonometric polynomials}
\label{sec:discrete_sp}
\begin{definition}[Discrete spaces]
Using discretization operator $\IM{\VM}$ and DFT matrix $\MBF$ defined in \eqref{eq:IN} and \eqref{eq:DFT} respectively, the discrete spaces are introduced as
\begin{subequations}
\label{eq:discrete_sp}
\begin{align}
\xUNM &= \IM{\VM}[\cU],
&
\xENM &= \IM{\VM}[\cEN],
&
\xJNM &= \IM{\VM}[\cJN],
\\
\xUNF &= \MBF_\VN[\xUNN],
&
\xENF &= \MBF_\VN[\xENN],
&
\xJNF &= \MBF_\VN[\xJNN].
\end{align}
\end{subequations}
\end{definition}
Thanks to the operator $\IM{\VM}$ being an isometric isomorphism, see Lemma~\ref{lem:IN_iso}, Helmholtz decomposition \eqref{eq:Helmholtz_trig} for trigonometric polynomials is transformed to discrete spaces
\begin{align*}
\IM{\VN}[\cTNd] &= \xUNN \oplus \xENN \oplus \xJNN = \xRdN,
&
\IM{\VM}[\cTNd] &= \xUNM \oplus \xENM \oplus \xJNM \subset \xRdM.
\end{align*}
Now, the discrete projections on subspaces \eqref{eq:discrete_sp} are defined; for better orientation among operators and subspaces, see following diagram
\begin{align}
\label{eq:oper_diagram}
  \begin{array}{ccccccccc}
\xENF & \stackrel{\MBF_\VN}{\longleftarrow} & \xENN & \stackrel{\IM{\VN}}{\longleftarrow} & \cEN & \stackrel{\IM{\VM}}{\longrightarrow} & \xENM & &
 \\  
 \rotatebox{90}{$\stackrel{\MBhGNN{\cE}}{\longrightarrow}$} &  & \rotatebox{90}{$\stackrel{\MBGNN{\cE}}{\longrightarrow}$} &  & \rotatebox{90}{$\stackrel{\mathcal{G}\sub{\cE}}{\longrightarrow}$} &  & \rotatebox{90}{$\stackrel{\MBGNM{\cE}}{\longrightarrow}$} & &
\\
\MBF_\VN[\xRdN] & \stackrel{\MBF_\VN}{\longleftarrow} & \xRdN & \stackrel{\IM{\VN}}{\longleftarrow} & \cTNd & \stackrel{\IM{\VM}}{\longrightarrow} & \IM{\VM}[\cTNd] & \subseteq & \xRdM
\\
\rotatebox{90}{$\stackrel{\MBhGNN{\cJ}}{\longleftarrow}$}  &  & \rotatebox{90}{$\stackrel{\MBGNN{\cJ}}{\longleftarrow}$} &  & \rotatebox{90}{$\stackrel{\mathcal{G}\sub{\cJ}}{\longleftarrow}$} &  & \rotatebox{90}{$\stackrel{\MBGNM{\cJ}}{\longleftarrow}$} & &
\\
\xJNF & \stackrel{\MBF_\VN}{\longleftarrow} & \xJNN & \stackrel{\IM{\VN}}{\longleftarrow} & \cJN & \stackrel{\IM{\VM}}{\longrightarrow} & \xJNM & &
  \end{array}.
\end{align}
\begin{definition}[Discrete projections]
\label{def:discrete_proj}
Let $\VN\in\xRd$ satisfy the odd grid assumption \eqref{eq:N_odd}, $\VM\in\xRd$ be a vector such that $M_\alp\geq N_\alp$ for all $\alp$.
Then for $\bullet\in\{\cU,\cE,\cJ\}$, matrices $ \MBhGNM{\bullet},\MBGNM{\bullet}\in\xMM$ are defined as
\begin{align*}
 \Bigl(\MBhGNM{\bullet}\Bigr)^{\Vk\Vl} &= 
 \begin{cases}
  \mhG^\bullet(\Vk)\del_{\Vk\Vl},
  &
  \text{for }\Vk,\Vl\in\ZNd
  \\
  \Vo
  &
  \text{for }\Vk,\Vl\in\ZMd\setminus\ZNd
 \end{cases},
 &
 \MBGNM{\bullet} &= \MBFi_\VM \MBhGNM{\bullet} \MBF_\VM.
\end{align*}
where the matrices $\mhG^\bullet(\Vk)\in\xRdd$ are Fourier coefficients of continuous projections introduced in Definition \ref{def:projection}, and where $\MBF_\VM$ is the DFT matrix from \eqref{eq:DFT}.
\end{definition}
\begin{lemma}[Discrete projections]
For $\bullet\in\{\cU,\cE,\cJ\}$,
\begin{enumerate}
\item operators $\MBhG_{\VN,\VN}^{\bullet}:\xCdN\rightarrow\xCdN$ are orthogonal projections on $\xUNF,\xENF$, and $\xJNF$,
\item operators
 $\MBGNM{\bullet}:\xRdM\rightarrow\xRdM$ are orthogonal projections on $\xUNM,\xENM$, and $\xJNM$.
\end{enumerate}
\end{lemma}
\begin{proof}
The fact that operators $\MBhG_{\VN,\VN}^{\bullet}$ and  $\MBGNM{\bullet}$ are mutually orthogonal projections follows from direct calculation; the properties are inherited from continuous projections in Definition~\ref{def:projection}. The images of individual projections follow from the properties of operators $\IM{\VM}$, $\MBF_\VN$, $\mathcal{G}\sub{\bullet}$
along with the definition of subspaces \eqref{eq:discrete_sp} and \eqref{eq:Helmholtz_trig}. Indeed, the discrete projections can be expressed as
\begin{align*}
\MBhG_{\VN,\VN}^{\bullet}\FT{\MBv}_\VN &= \MBF_\VN\circ\IM{\VN}\circ\mathcal{G}\sub{\bullet}\circ\IMi{\VN}\circ\MBFi_\VN\FT{\MBv}_\VN,
& 
\MBGNM{\bullet}\MBv_\VN &= \IM{\VM}\circ\mathcal{G}\sub{\bullet}\circ\IMi{\VM}\MBv_\VN
\end{align*}
for $\FT{\MBv}_\VN\in\MBF_\VN[\xRdN]$ and $\MBv_\VN\in\IM{\VM}[\cTNd]$; see the structure of operators and subspaces \eqref{eq:oper_diagram}.
\end{proof}
\subsection{Linear systems}
\label{sec:linear_system}
This section deals with resolutions of discrete minimizers from linear systems. This topic has already been studied in \cite[section~5]{VoZeMa2014FFTH} and \cite[section~7]{VoZeMa2014GarBounds} for the GaNi scheme \eqref{eq:GaNi}.
Here, the concept is summarized and extended to the Ga scheme \eqref{eq:Ga}.
\begin{proposition}[From minimization to linear system]
\label{lem:min2linear_sys}
Let $\cH$ be a Hilbert space with a nontrivial orthogonal
decomposition $\cH = \mathring{\cU} \oplus \mathring{\cE} \oplus \mathring{\cJ}$, where $\mathring{\cU}$ is isometrically isomorphic with $\xRd$. Next, let bilinear form $\bilfG{}{}:\cH\times\cH\rightarrow\xR$ be defined as $\bilfG{\Vu}{\Vv} = \scal{\rTA\Vu}{\Vv}_{\cH}$,
for the symmetric, coercive, and bounded linear operator $\rTA:\cH\rightarrow\cH$, i.e. there exist $c_{\rTA}>0$ and $C_{\rTA}>0$ such that $c_{\rTA} \|\Vu\|_{\cH} \leq \scal{\rTA\Vu}{\Vu}_{\cH} \leq C_{\rTA} \|\Vu\|_{\cH}$ for all $\Vu\in\cH$.
Then a problem for $\VE\in\rcU$ to find a minimizer $\rVe\mac{\VE}\in\rcE$ of
\begin{subequations}
\begin{align}\label{eq:q2l_minimization}
  \rVe\mac{\VE} = \argmin_{\rVe\in\rcE} \bilfG{\VE + \rVe}{\VE + \rVe}
\end{align}
is equivalent to finding the solution $\rVe\mac{\VE}\in\rcE$ of the following equation in $\cH$
\begin{align}\label{eq:q2l_equation}
 \mathring{\T{G}}\rTA\rVe\mac{\VE} = - \mathring{\T{G}}\rTA\VE,
\end{align}
\end{subequations}
where $\mathring{\T{G}}$ is an orthogonal projection on $\rcE$.
\end{proposition}
\begin{proof}
The proof starts with an optimality condition of
\eqref{eq:q2l_minimization}, namely
$\bilfG{\rVe\mac{\VE}_\VN}{\Vv} = -\bilfG{\VE}{\Vv}$ for all $\Vv\in\rcE$.
Then, the projection is incorporated in order to enlarge the space of test functions
\begin{align*}
\bilfG{\rVe\mac{\VE}}{\mathring{\T{G}}\Vv} &= -\bilfG{\VE}{\mathring{\T{G}}\Vv}
\quad\forall\Vv\in\cH,
\\
\scal{\mathring{\T{G}}\rTA\rVe\mac{\VE}_\VN}{\Vv}_{\cH} &= -\scal{\mathring{\T{G}}\rTA\VE}{\Vv}_{\cH}
\quad\forall\Vv\in\cH,
\end{align*}
where the orthogonality (symmetry) of 
$\mathring{\T{G}}$ has also been used.
Now, it is possible to remove the scalar product and deduce the required \eqref{eq:q2l_equation}.
\end{proof}
\begin{remark}[Linear systems for the GaNi]
\label{lem:GaNi_linsys}
According to \cite[Proposition~12]{VoZeMa2014FFTH}, with the notation from Remark~\ref{rem:FD_GaNi},
the minimizers $\IM{\VN}[\gVe_\VN\mac{\VE}]=\gMBe_{\VN}\mac{\VE}\in\xENN$ and $\IM{\VN}[\gVj_\VN\mac{\VJ}]=\gMBj_{\VN}\mac{\VJ}\in\xJNN$ for $\VE,\VJ\in\xRd$ in the GaNi scheme \eqref{eq:GaNi} satisfy the following equations
\begin{align}
\label{eq:GaNi_linsys}
\MBGNN{\cE}\MBtA_\VN \gMBe\mac{\VE}_\VN &= -\MBGNN{\cE}\MBtA_\VN\VE,
&
\MBGNN{\cJ}\MBtA_\VN^{-1} \gMBj\mac{\VJ}_\VN &= -\MBGNN{\cJ}\MBtA_\VN^{-1}\VJ,
\end{align}
where $\MBGNN{\cE}$ and $\MBGNN{\cJ}$ are projection matrices from Def.~\ref{def:discrete_proj} and $\MBtA_\VN^{\Vk\Vm}=\del_{\Vk\Vm}\FT{\TA}(\Vk)$ for $\Vk,\Vm\in\ZNd$.
\end{remark}
\dlt{\begin{remark}[Pocet FFT]
reseni v $\xC$:
\begin{align}
\MBhGNN{\cE}2^{-d} \sum_{\Vs\in\xSd} \MBS_\VN^*(\Vs) 
\MBF_\VN 
\MBA_{\VN}
(\Vs)\MBFi_\VN\MBS_\VN(\Vs)
\end{align}
pocet FFT: $2^{d+1}$, tj. pro $d=2,3$ dostavame: $8, 16$
\\
reseni v $\xR$:
\begin{align}
\MBFi_\VN \MBGNN{\cE} 2^{-d}\left(\MBF_\VN 
\MBA_{\VN}
(\Vs) + \sum_{\Vs\in\xSd\setminus\{\Vo\}}  \MBS_\VN^*(\Vs) 
\MBF_\VN 
\MBA_{\VN}
(\Vs)\MBFi_\VN\MBS_\VN(\Vs) \MBF_\VN \right),
\end{align}
pocet FFT: $3\cdot (2^{d} - 1) + 2$, tj. pro $d=2,3$:  dostavame: $11, 23$
\end{remark}}

\begin{corollary}[Linear systems for the Ga]
\label{lem:Ga_linsys}
Let $\MBGNM{\bullet},\MBhGNM{\bullet}$
for $\bullet\in\{\cE,\cJ\}$ 
be projection matrices from Definition~\ref{def:discrete_proj}.
Then, for the minimizers $\Ve_\VN\mac{\VE}\in\cEN$ and $\Vj_\VN\mac{\VJ}\in\cJN$ of the Ga scheme \eqref{eq:Ga}  for $\VE,\VJ\in\xRd$, the following hold:
\begin{enumerate}
 \item For $\VM=\tVNr$, the minimizers $\MBe_\VN\mac{\VE}:=\IM{\VM}[\Ve_\VN\mac{\VE}]\in\xENM$ and $\MBj_\VN\mac{\VJ}:=\IM{\VM}[\Vj_\VN\mac{\VJ}]\in\xJNM$ satisfy
\begin{subequations}
\label{eq:Ga_linsys}
\begin{align}
\label{eq:Ga_linsys_orig}
\MBGNM{\cE} \MBA_\VM \MBe\mac{\VE}_\VN &= -\MBGNM{\cE} \MBA_\VM\VE,
&
\MBGNM{\cJ} \MBB_\VM \MBj\mac{\VJ}_\VN &= -\MBGNM{\cJ} \MBB_\VM\VJ,
\end{align}
where $\MBA_\VM, \MBB_\VM\in\xMM$ are defined in \eqref{eq:MBA_GA}, particularly in \eqref{eq:inclcomp_formula} and \eqref{eq:gridcomp_formula} for an inclusion-based \eqref{eq:inclcomp} and grid-based \eqref{eq:gridcomp} composite.
\item The minimizers $\FT{\MBe}_\VN\mac{\VE}:=\MBF_\VN\IM{\VN}[\Ve_\VN\mac{\VE}]\in\xENF$ and $\FT{\MBj}_\VN\mac{\VJ}:=\MBF_\VN\IM{\VN}[\Vj_\VN\mac{\VJ}]\in\xJNF$ satisfy
\begin{align}
\label{eq:Ga_linsys_R}
\MBhGNN{\cE} \MBhA_\VN^\full \FT{\MBe}\mac{\VE}_\VN &= -\MBhGNN{\cE} \MBhA_\VN^\full\FT{\VE},
&
\MBhGNN{\cJ}\FT{\MBB}_\VN^\full \FT{\MBj}\mac{\VJ}_\VN &= -\MBhGNN{\cJ}\FT{\MBB}_\VN^\full\FT{\VJ},
\end{align}
\end{subequations}
where matrices $\MBhA_\VN^\full$, $\FT{\MBB}_\VN^\full\in\xhMN$ are defined according to sparse decomposition \eqref{eq:Afull_decomp_Fourier}.
\end{enumerate}
 \end{corollary}
\subsection{Computational and implementation issues}
\label{sec:computational_aspects}
Here, practical aspects regarding the resolution of minimizers from linear systems are discussed.
\begin{remark}[Solution by conjugate gradients]
\label{rem:CG}
The discrete problems, see the Ga \eqref{eq:Ga} and GaNi \eqref{eq:GaNi} schemes, can be effectively solved with Krylov subspace methods \cite{Trefethen1997NLA,Saad2003IMSL}, particularly conjugate gradients \cite[Algorithm~6.18]{Saad2003IMSL}. It was pointed out in \cite{ZeVoNoMa2010AFFTH,Brisard2010FFT} and explained by variational reformulation in \cite{VoZeMa2012LNSC,VoZeMa2014FFTH} for the GaNi scheme \eqref{eq:GaNi}.

Using the general notation from Proposition~\ref{lem:min2linear_sys}, the minimization problems of both discrete schemes \eqref{eq:Ga} and \eqref{eq:GaNi} rely on the quadratic functional \eqref{eq:q2l_minimization} with a symmetric and positive definite matrix.
Thus, conjugate gradients (CG) can be employed as the minimization over subspace $\rcE$ is carried out with projection operator $\mathring{\T{G}}$.

The minimization process also corresponds to the solution of the linear system \eqref{eq:q2l_equation} with an initial approximation $\rMBe_{(0)}\mac{\VE}$ from the minimization space $\rcE$, which ensures that a residual vector
\begin{align*}
 \MB{r}\iter{k} = -\mathring{\T{G}}\rTA\rVe\mac{\VE}\iter{k} - \mathring{\T{G}}\rTA\VE
\end{align*}
is from the subspace $\rcE$ for arbitrary $k$-th iteration.
Then, the CG algorithm is interpreted as a minimization
\begin{align}
\label{eq:minim_krylov}
  \rVe\mac{\VE}\iter{k} = \argmin_{\rVe\in\rVe\mac{\VE}\iter{0}+\set{K}_{(i)}} \bilfG{\VE + \rVe}{\VE + \rVe}
\end{align}
over Krylov subspaces defined for $i=1,2,\dotsc$ as
\begin{align*}
\set{K}_{(i)} 
=
\mathrm{span}
\Bigl\{ \MB{r}\iter{0}, 
  \mathring{\T{G}} \rTA
  \MB{r}\iter{0}, 
  \ldots, 
  (\mathring{\T{G}} \rTA)^{i-1} \MB{r}\iter{0} 
\Bigr\}
\quad\text{satisfying }
\set{K}_{(i)} \subseteq \set{K}_{(i+1)} \subseteq\rcE=\mathring{\T{G}}[\cH].
\end{align*}
The application of the CG algorithm only requires the implementation of the matrix-vector multiplication of the linear system.  For the GaNi \eqref{eq:GaNi_linsys} and the Ga \eqref{eq:Ga_linsys}, it is outlined in Algorithms~\ref{alg:GaNi}, \ref{alg:Ga_orig}, and \ref{alg:Ga_reduced}.
\end{remark}
\begin{algorithm}
\caption{
Matrix-vector multiplication for the primal formulation in the GaNi \eqref{eq:GaNi_linsys}}
\label{alg:GaNi}
\begin{algorithmic}[1]
\Require
{$\MBA\in\xR^{d\times d\times\VN}\leftarrow$ nonzero elements of  $\MBtA_\VN\in\xMN$
}
\Comment{$\M{A}_{\alp,\beta}^{\Vk}=A_{\alp\beta}(\Vx_\VN^\Vk)$ for $\alp,\beta=1,\dotsc,d$ and $\Vk\in\ZNd$; optionally $\MBhG\in\xR^{d\times d\times\VN}\leftarrow$ nonzero elements of $\MBhGNN{\cE}\in\xMN$}
\Procedure{multiplication}{$\MBA$, $\MBx$}\Comment{calculates $\MBy=\MBGNN{\cE}\MBtA_\VN\MBx\in\xRdN$ for $\MBx\in\xRdN$}
\State{$\MBy\leftarrow\MBtA_\VN\MBx$}
\Comment{$\My_{\alp}^{\Vk} = \sum_{\beta} \MA_{\alp\beta}^\Vk \Mx_{\beta}^\Vk$ for all $\alp$ and $\Vk\in\ZNd$}
\State{$\MBy\leftarrow \MBF_\VN\MBy$}
\Comment{$\MBy_\alp = \mathrm{FFT}_\VN(\MBy_\alp)$ for all $\alp$ with an FFT of size $\VN$}
\State{$\MBy\leftarrow\MBhGNN{\cE}\MBy$}
\Comment{$\My_{\alp}^{\Vk} = \sum_{\beta} \frac{k_\alp k_\beta}{\|\Vk\|_\xRd}
\My_{\beta}^\Vk$ for all $\alp$ and $\Vk\in\ZNd$}
\State{$\MBy\leftarrow \MBFi_\VN\MBy$}
\Comment{$\MBy_\alp = \mathrm{iFFT}_\VN(\MBy_\alp)$ for all $\alp$ with an inverse FFT of size $\VN$}
%
\EndProcedure{\textbf{: return }$\MBy$}
\end{algorithmic}

\end{algorithm}
\begin{algorithm}
\caption{
Matrix-vector multiplication for the primal formulation in the double grid Ga \eqref{eq:Ga_linsys_orig}}
\label{alg:Ga_orig}
\begin{algorithmic}[1]
\Require
{$\MBA\in\xR^{d\times d\times\VM}\leftarrow$ nonzero elements of $\MBA_\VM\in\xMM$ for $\VM=\tVNr$}
\Comment{for an evaluation, see Remark~\ref{rem:eval_with_FFT};
optionally $\MBhG\in\xR^{d\times d\times\VN}\leftarrow$ nonzero elements of $\MBhGNM{\cE}\in\xMM$}
\Procedure{multiplication}{$\MBA$, $\MBx$}\Comment{calculates $\MBy=\MBGNM{\cE}\MBA_\VM\MBx\in\xRdM$ for $\MBx\in\xRdM$}
\State{$\MBy\leftarrow\MBA_\VM\MBx$}
\Comment{$\My_{\alp}^{\Vk} = \sum_{\beta} \MA_{\alp\beta}^\Vk \Mx_{\beta}^\Vk$ for $\alp=1,\dotsc,d$ and  $\Vk\in\ZMd$}
\State{$\MBy\leftarrow\MBF_\VM\MBy$}
\Comment{$\MBy_\alp = \mathrm{FFT}_\VM(\MBy_\alp)$ for all $\alp$ with an FFT of size $\VM$}
\State{$\MBy\leftarrow\MBhGNM{\cE}\MBy$}
\Comment{$\My_{\alp}^{\Vk} = \sum_{\beta} \frac{k_\alp k_\beta}{\|\Vk\|_\xRd}
\My_{\beta}^\Vk$ for $\Vk\in\ZNd$ and $\MBy_{\alp}^{\Vk}=0$ for $\Vk\in\ZMd\setminus\ZNd$ and all $\alp$}
\State{$\MBy\leftarrow \MBFi_\VM\MBy$}
\Comment{$\MBy_\alp = \mathrm{iFFT}_\VM(\MBy_\alp)$ for all $\alp$ with an inverse FFT of size $\VM$}
\EndProcedure{\textbf{: return }$\MBy$}
\end{algorithmic}
\end{algorithm}
\begin{algorithm}
\caption{Matrix-vector multiplication for the primal formulation in the reduced Ga \eqref{eq:Ga_linsys_R}}
\label{alg:Ga_reduced}
\begin{algorithmic}[1]
\Require
{$\MBA\in\xR^{d\times d\times\tVN}\leftarrow$ nonzero elements of $\MBA_\tVN\in\xMtN$}
\Comment{for an evaluation, see Remark~\ref{rem:eval_with_FFT};
optionally $\MBhG\in\xR^{d\times d\times\VN}\leftarrow$ nonzero elements of $\MBhGNN{\cE}\in\xMN$}
\Procedure{multiplication}{$\MBA$, $\MBx$}\Comment{calculates $\MBy=\MBhGNN{\cE}\MBhA^{\mathrm{full}}_\VN\MBx\in\xCdN$}
for $\MBx\in\xCdN$
\State{$\MBy\leftarrow\Vo\in\xCdN$}
\For{$\Vs\in \xS^d=\{0,1\}^d$}
\State{$\MB{z}\leftarrow\MBS_\VN(\Vs)\MBx$}
\Comment{$\M{z}_\alp^\Vk=\omega_\tVN^{-\Vs\Vk}\Mx_\alp^\Vk$ for all $\alp$ and $\Vk\in\ZNd$}
\State{$\MB{z}\leftarrow \MBFi_\VN\MB{z}$}
\Comment{$\MB{z}_\alp = \mathrm{iFFT}(\MB{z}_\alp)$ for all $\alp$ with an inverse FFT of size $\VN$}
\State{$\MB{z}\leftarrow\MBA_\VN(\Vs)\MB{z}$}
\Comment{$\M{z}_{\alp}^{\Vk} = \sum_{\beta} \MA_{\alp\beta}^{2\Vk-\Vs}\M{z}_{\beta}^{\Vk}$ for all $\alp$ and $\Vk\in\ZNd$}
\State{$\MB{z}\leftarrow \MBF_\VN\MB{z}$}
\Comment{$\MB{z}_\alp = \mathrm{FFT}(\MB{z}_\alp)$ for all $\alp$ with an FFT of size $\VN$}
\State{$\MBy\leftarrow\MBy+2^{-d}\MBS_\VN^*(\Vs)\MB{z}$}
\Comment{$\My_\alp^\Vk=\My_\alp^\Vk+2^{-d}\omega_\tVN^{\Vs\Vk}\M{z}_\alp^\Vk$ for all $\alp$ and $\Vk\in\ZNd$}
\EndFor
\State{$\MBy\leftarrow\MBhGNN{\cE}\MBy$}
\Comment{$\My_{\alp}^{\Vk} = \sum_{\beta} \frac{k_\alp k_\beta}{\|\Vk\|_\xRd}
\My_{\beta}^\Vk$ for all $\alp$ and $\Vk\in\ZNd$}
\EndProcedure{\textbf{: return }$\MBy$}
\end{algorithmic}
\end{algorithm}
\begin{remark}[Memory and computational requirements]
\label{rem:mem_req}
For the approximation order $\VN$ of trigonometric polynomials, the linear systems in \eqref{eq:GaNi_linsys} and \eqref{eq:Ga_linsys} have different sizes leading to different memory requirements, see~Table~\ref{tab:memory}.
\JV{Noting that memory demands can be further reduced by incorporating the symmetry of coef. matrix $\MBA$ and by calculating projection matrix $\MBhG$ instead of storing, when needed. 
Despite the different sizes of linear systems,} the number of independent unknowns remains the same and is equal to the dimension of approximation spaces $\cEN,\cJN$ or their discrete relatives $\xENN,\xJNN$.
The subspace for primal formulations has a dimension $\dim\xEN=\pVN-1$ since this can be expressed using potential with zero-mean. Because the dimensions of constant fields and the whole space are $\dim\xUN=d$ and $\dim\xRdN=d\pVN$, the dimension of the dual space is equal to $\dim\xJN = (d-1)(\pVN-1)$.
\\
The linear systems for GaNi \eqref{eq:GaNi_linsys} and for Ga \eqref{eq:Ga_linsys_orig} possess exactly the same mathematical structure with a block-diagonal matrix of material coefficients (for isotropic material, only diagonal); see Algorithms~\ref{alg:GaNi} and~\ref{alg:Ga_orig}; compare Remark~\ref{rem:FD_GaNi} with
Lemma~\ref{lem:integ_eval_VM}. However, the Ga has a double size of vectors and matrices in the linear system. The corresponding higher memory and computational requirements are outperformed with higher accuracy for the Ga scheme, see section~\ref{sec:numerical_errors_for_phase_contrast} for a comparison.
\\
In accordance with Lemma~\ref{lem:reduce_grid}, the reduced Ga scheme \eqref{eq:Ga_linsys_R} benefits from the size reduction of an unknown vector, which is amplified when more vectors are stored (conjugate gradients, nonlinear problems and solvers, etc.). Furthermore, the computational requirements remain approximately the same, which is illustrated in Figure~\ref{fig:Ax_speed}.
\begin{table}
\caption{Memory requirements (no. of components stored in linear systems for anisotropic material coef.)}
\centering
\begin{tabular}{llll}
\toprule
 & GaNi \eqref{eq:GaNi_linsys} & Ga \eqref{eq:Ga_linsys_orig} & Ga reduced \eqref{eq:Ga_linsys_R} \\
\midrule
unknowns/right-hand side & $d\pVN$ & $d\ptVNr\approx d2^d\pVN$ & $d\pVN$ \\
matrix of material coef. $\MBA$ & $d^2\pVN$ & $d^2\ptVNr\approx d^22^d\pVN$ & $d^2\ptVN= d^22^d\pVN$ 
\\
projection matrix $\MBhG$ & $d^2\pVN$ & $d^2\pVN$ & $d^2\pVN$ 
\\
\midrule
\JV{
indepen. unknowns in primal form.} & $\pVN-1$ & $\pVN-1$ & $\pVN-1$ \\
\JV{indepen. unknowns in dual form.} & $(d-1)(\pVN-1)$ & $(d-1)(\pVN-1)$ & $(d-1)(\pVN-1)$ \\
\bottomrule

\end{tabular}
\label{tab:memory}
\end{table}
\end{remark}
\begin{figure}[htp]
\centering
\subfigure{
\includegraphics[scale=0.6]{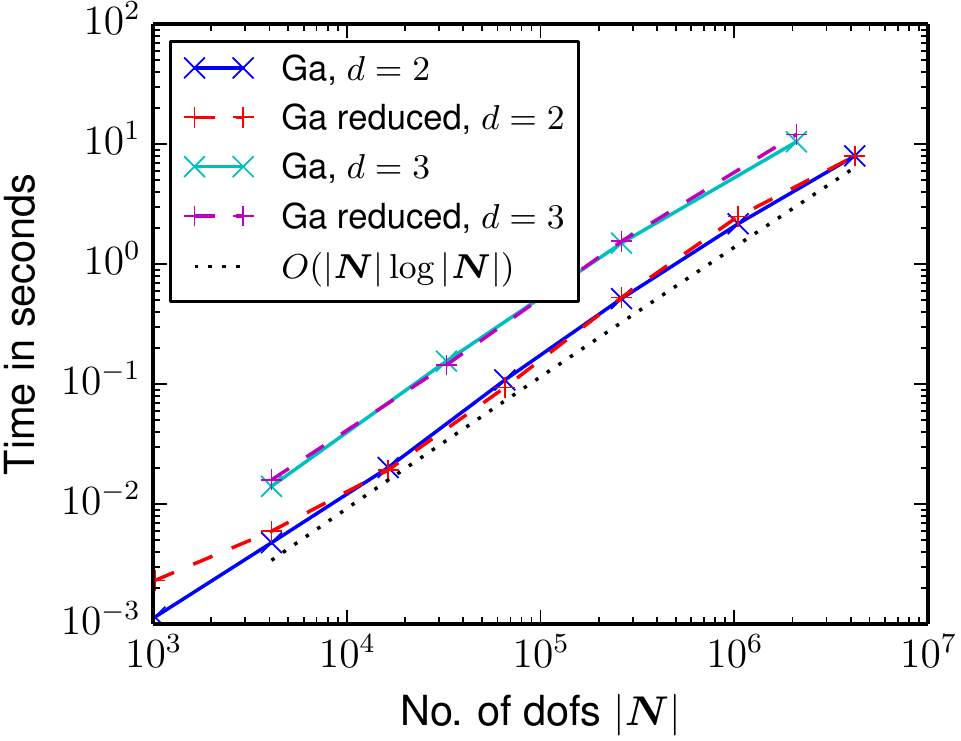}
}
\subfigure{
\includegraphics[scale=0.6]{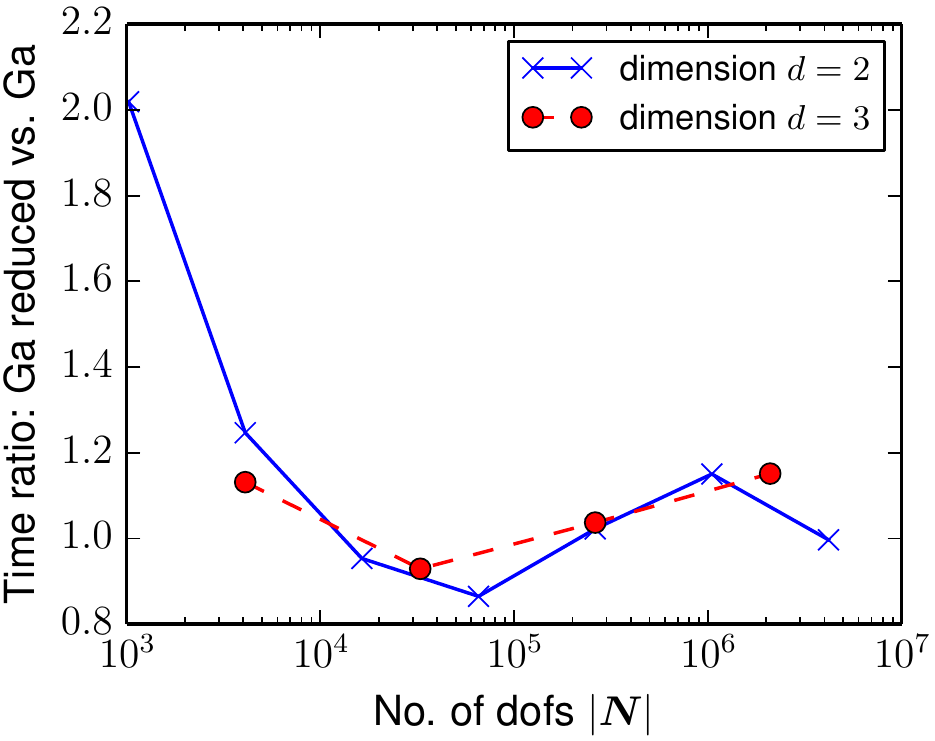}
}
\caption{Comparison of matrix-vector multiplication for Ga \eqref{eq:Ga_linsys_orig} with its reduced version \eqref{eq:Ga_linsys_R}}
\label{fig:Ax_speed}
\end{figure}
\begin{remark}[Evaluation of material coefficients matrices with FFT]
\label{rem:eval_with_FFT}
The matrix \eqref{eq:gridcomp_formula} derived for the grid-based composite \eqref{eq:gridcomp} can be evaluated efficiently using the FFT algorithm; for the inclusion-matrix composite \eqref{eq:inclcomp}, the effective evaluation of \eqref{eq:inclcomp_formula} was discussed in \cite[Remark~48]{VoZeMa2014FFTH}.

In \eqref{eq:gridcomp_formula}, the sum over $\ZPd$ and $\ZtNrd$ is provided by $\VP$-sized FFT and $\VM$-sized inverse FFT algorithm resp., whereas the factor $\FT{\psi}(\Vm)$ for $\Vm\in\ZtNrd$ occurs as an element-wise multiplication. However, for $\VP\neq \tVNr$, the additional treatment has to be provided. For $P_\alp>2N_\alp-1$, the vector
\begin{align*}
\Bigl( \sum_{\Vp\in \ZPd} \frac{\omega_{\VP}^{-\Vm\Vp}}{\pVP} 
 \M{C}_{\VP,\alp\beta}^{\Vp} \Bigr)^{\Vm\in\ZPd}\in\xC^{\VP}
\end{align*}
is truncated to $\xC^{2\VN-1}$, while for $P_\alp<2N_\alp-1$, it is periodically enlarged to $\xC^{2\VN-1}$ thanks to the periodicity of $\omega_\VP^{\cdot\Vp}$.
\end{remark}

\section{Numerical examples}
\label{sec:numerical_examples}
This section is dedicated to numerical examples that confirm the properties of guaranteed bounds \eqref{eq:structure_homog_matrices} with an emphasis on the comparison of Ga \eqref{eq:Ga} with GaNi  \eqref{eq:GaNi} and \eqref{eq:GaNi_apost}.

\begin{problem}
A two-dimensional problem with material coefficients defined on a periodic cell $\puc = (-1,1)\times(-1,1) \subset \xR^2$ is considered and defined via
\begin{align*}
\TA(\Vx) = \mI[1 + \rho f_{\bullet}(\Vx)]\quad\text{for }\Vx\in\puc,
\end{align*}
where $\mI\in\xR^{2\times 2}$ is the identity matrix, $f: \puc\rightarrow \xR$ is a scalar nonnegative function which controls the shape of inclusions (recall Remark~\ref{rem:anal_Fourier} for specific examples), and $\rho>0$ is a parameter corresponding to the phase contrast. Two types of inclusions, square and circle, are considered, namely
\begin{align}
\label{eq:incl_topologies}
f_\mathrm{square}(\Vx) &= 
\begin{cases}
1&\text{for }\|\Vx\|_{\infty}< s \\
0&\text{otherwise}
\end{cases},
&
f_\mathrm{circle}(\Vx) &= 
\begin{cases}
1&\text{for }\|\Vx\|_{2}<s \\
0&\text{otherwise}
\end{cases},
\end{align}
where parameter $2s$ corresponds to an inclusion size, the side of the square and the radius, respectively.
The problem is discretized with odd grids \eqref{eq:N_odd} with an example shown in Figure~\ref{fig:cells} along with inclusion interfaces for both geometries \eqref{eq:incl_topologies}.
\end{problem}
\begin{figure}[htp]
\centering
\subfigure[\small Square inclusion]{
\includegraphics[scale=0.55]{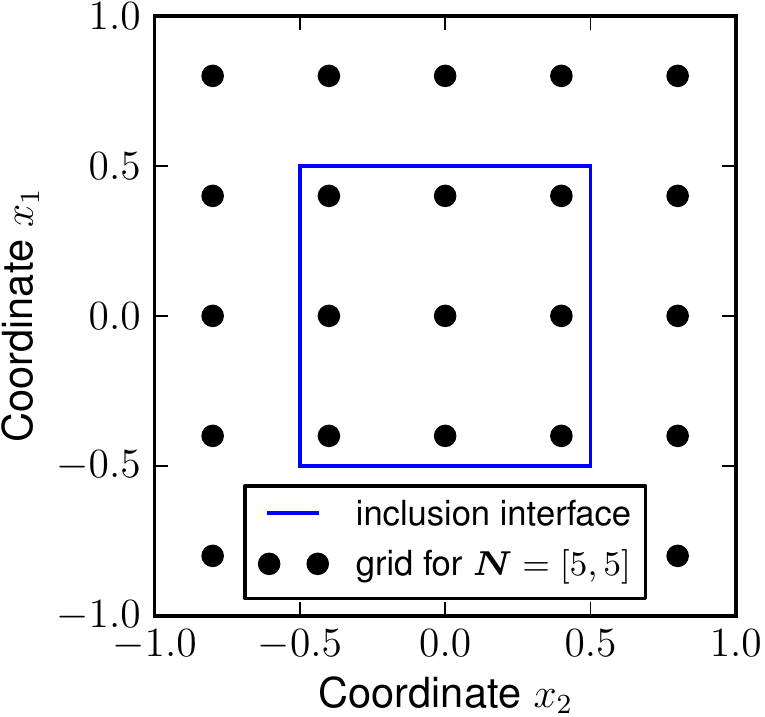}
}
\subfigure[\small Circle inclusion]{
\includegraphics[scale=0.55]{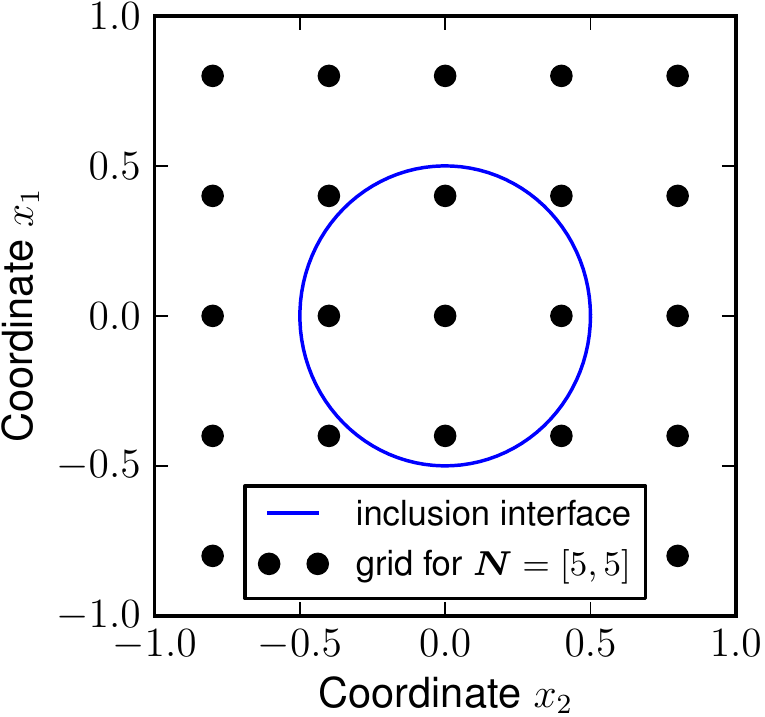}
}
\caption{Cells with grid and inclusion interfaces for size $s=\frac{1}{2}$}
\label{fig:cells}
\end{figure}
\begin{remark} All the computations have been provided using Python software \texttt{FFTHomPy} available at: \url{https://github.com/vondrejc/FFTHomPy.git}.
The linear systems presented in section~\ref{sec:linear_system} have been solved by
conjugate gradients; a convergence criterion on the norm of residuum has been chosen with a relatively small tolerance $10^{-6}$ in order to suppress algebraic error.
\end{remark}

The numerical examples are separated into the following parts:
section~\ref{sec:numerical_sensitivity2size} explores sensitivity of homogenized properties in regard to inclusion size, 
section~\ref{sec:numerical_bounds_wrt_dofs} describes an evolution of upper-lower bounds for an increase in grid points, section~\ref{sec:numerical_errors_for_phase_contrast} treats the behavior with different phase contrasts, and section~\ref{sec:numerical_CG} shows the progress of guaranteed bounds during iterations of conjugate gradients.

\subsection{Numerical sensitivity for the inclusion size}
\label{sec:numerical_sensitivity2size}
Here, homogenized properties are investigated with regard to an inclusion size $s$. Figure~\ref{fig:sensitivity_homog} depicts the results for a relatively small number of discretization points $\VN=(5,5)$ which highlight the difference between the Ga \eqref{eq:Ga} and GaNi \eqref{eq:GaNi} schemes.

\begin{figure}[htp]
\centering
\subfigure[\small Square inclusion]{
\includegraphics[scale=0.6]{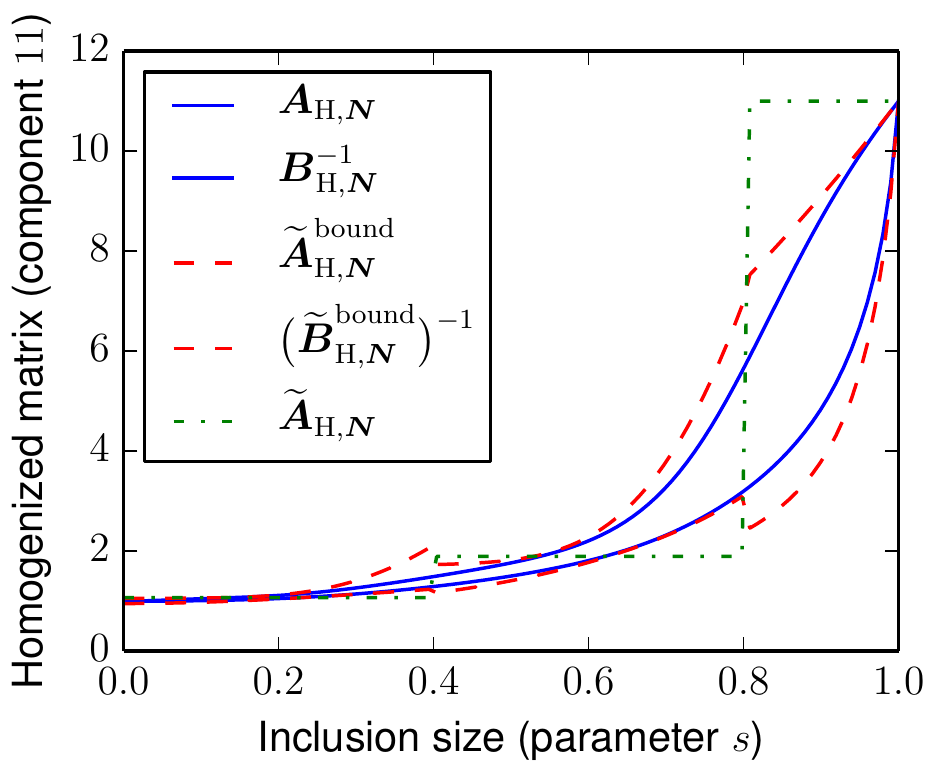}
}
\subfigure[\small Circle inclusion]{
\includegraphics[scale=0.6]{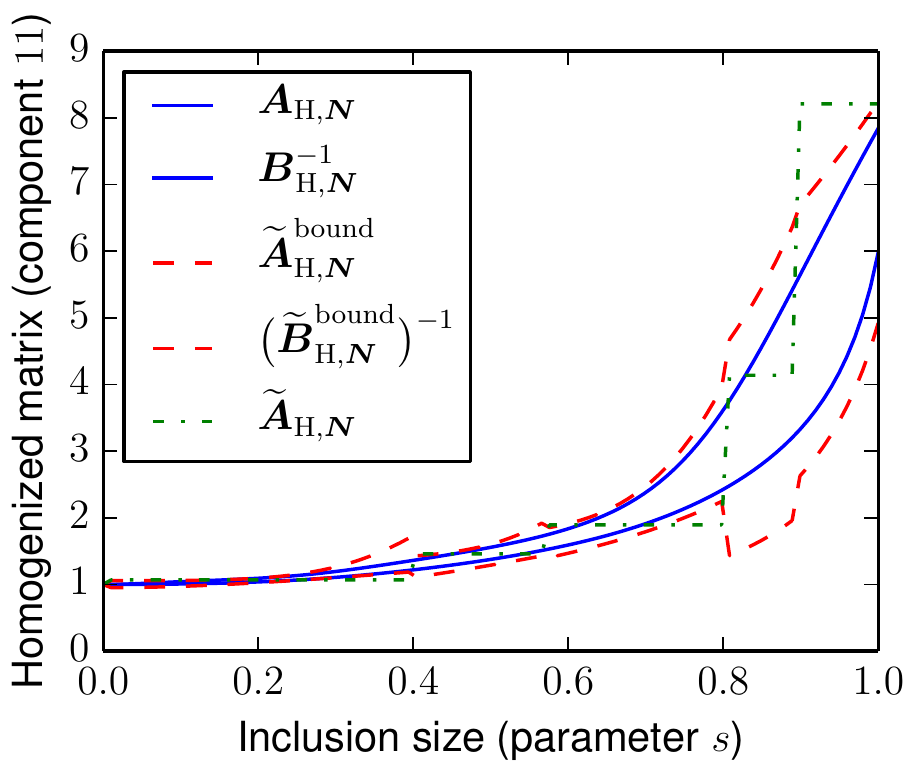}
}
\caption{Sensitivity of homogenized properties for $\rho = 10$ and  $\VN=(5,5)$}
\label{fig:sensitivity_homog}
\end{figure}

The structure in \eqref{eq:structure_homog_matrices} is satisfied, with better results obtained for homogenized coefficients $\AeffN,\BeffN^{-1}$ using the Ga scheme \eqref{eq:Ga} when compared to the guaranteed bounds $\tAeffN\bound,\bigl(\tBeffN\bound\bigr)^{-1}$ of the GaNi \eqref{eq:GaNi_apost}.
The GaNi matrix $\tAeffN$ in \eqref{eq:GaNi} together with its guaranteed bounds \eqref{eq:GaNi_apost} has already been studied in \cite{VoZeMa2014GarBounds}, where the authors pointed out that the homogenized matrix $\tAeffN$, in some cases, underestimates or overestimates its own guaranteed bounds $\tAeffN\bound$ and $\bigl(\tBeffN\bound\bigr)^{-1}$, respectively.
The GaNi scheme \eqref{eq:GaNi} is influenced by inaccurate numerical integration which disregards exact inclusion shapes because the scheme is defined only on grid points. 
As a result of exact integration, the homogenized matrices $\AeffN$, $\BeffN$ change smoothly in relation to the inclusion size $s$.

\subsection{Upper-lower bounds for an increase in the number of grid points}
\label{sec:numerical_bounds_wrt_dofs}
This section is dedicated to the behavior of homogenized properties for an increase in the approximation order of trigonometric polynomials $\VN$, see Definition~\ref{def:trig_pol}. It is depicted in Figures~\ref{fig:bounds_for_dofs} and~\ref{fig:errors_for_dofs} for homogenized properties and also for their guaranteed errors defined as 
\begin{align}
\label{eq:error}
\eta_\VN
&:= \tr\biggl(\frac{\AeffN-\BeffN^{-1}}{2}\biggr),
&
\widetilde{\eta}_\VN\bound &:= \tr\biggl( \frac{\tAeffN\bound-\bigl(\tBeffN\bound\bigr)^{-1}}{2} \biggr).
\end{align}
\begin{figure}[htp]
\centering
\subfigure[\small Square inclusion]{
\includegraphics[scale=0.6]{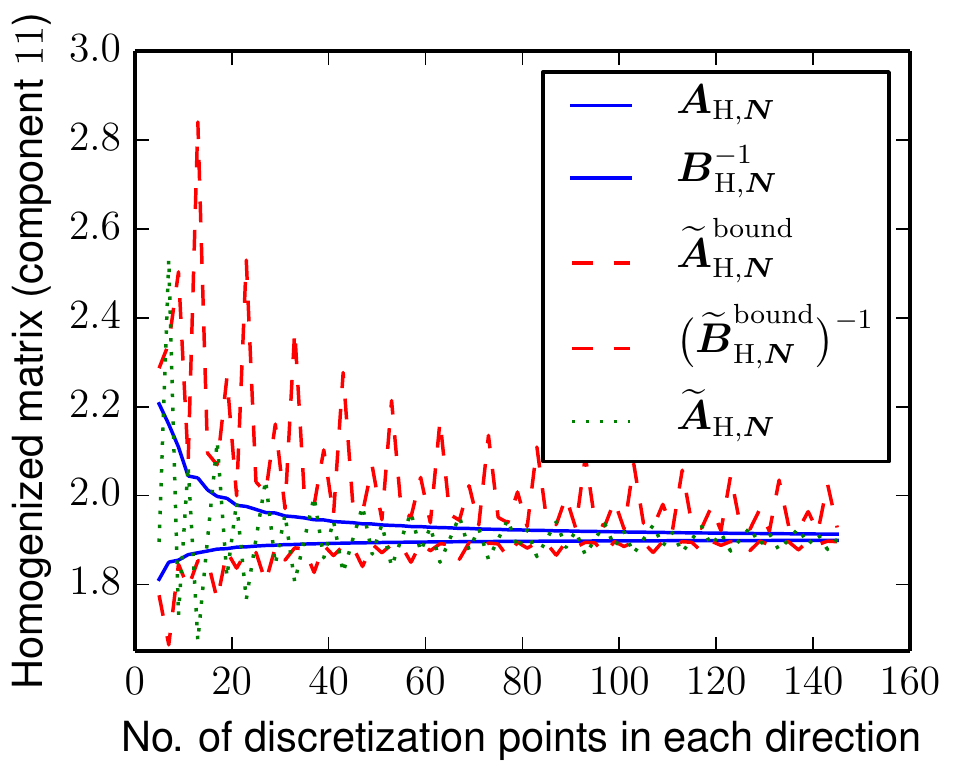}
}
\subfigure[\small Circle inclusion]{
\includegraphics[scale=0.6]{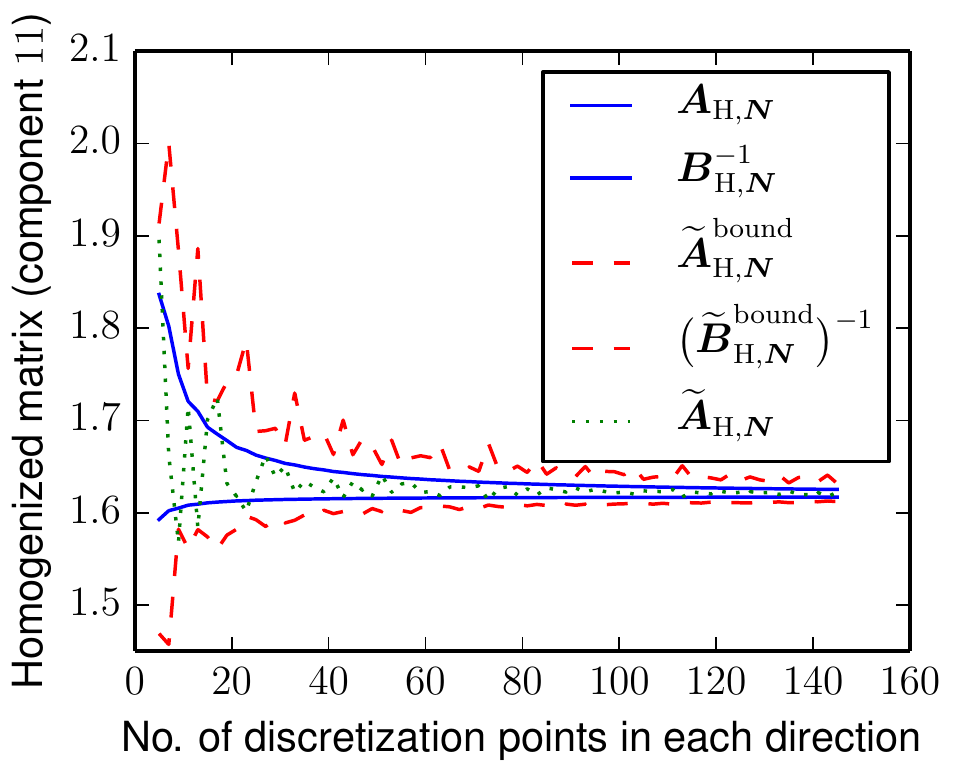}
}
\caption{Bounds on homogenized matrix for inclusion size $s=0.6$ and phase contrast $\rho = 10$}
\label{fig:bounds_for_dofs}
\end{figure}

All the homogenized coefficients from both the Ga and GaNi schemes support the structure of guaranteed bounds \eqref{eq:structure_homog_matrices} and converge to homogenized matrix $\Aeff$ for an increasing number of grid points, which has been proven theoretically in \cite[section~4.2]{VoZeMa2014FFTH} for the Ga scheme; the convergence for GaNi is provided in \cite[section~4.3]{VoZeMa2014FFTH} along with a regularization for discontinuous material coefficients according to \cite[Section~3, pp.~115--117]{Vondrejc2013PhD} or later in \cite{Schneider2014convergence} for Riemann integrable coefficients.
Moreover, thanks to the hierarchy of approximation spaces
\begin{align}
\label{eq:subspace_hierarchy}
\cEN \subseteq\cE_{\VM}\subset\cE\text{ and }\cJN \subseteq\cJ_{\VM}\subset\cJ\quad\text{for }N_\alp\leq M_\alp,
\end{align}
the homogenized matrices of the Ga scheme $\AeffN, \BeffN^{-1}$
evolve monotonically as opposed to the homogenized matrices of the GaNi scheme, which suffer, as already noticed in previous section, from inexact numerical integration causing the so-called ''variational crime'' \cite{strang1972varcrime}.

\begin{figure}[htp]
\centering
\subfigure[\small Square inclusion]{
\includegraphics[scale=0.6]{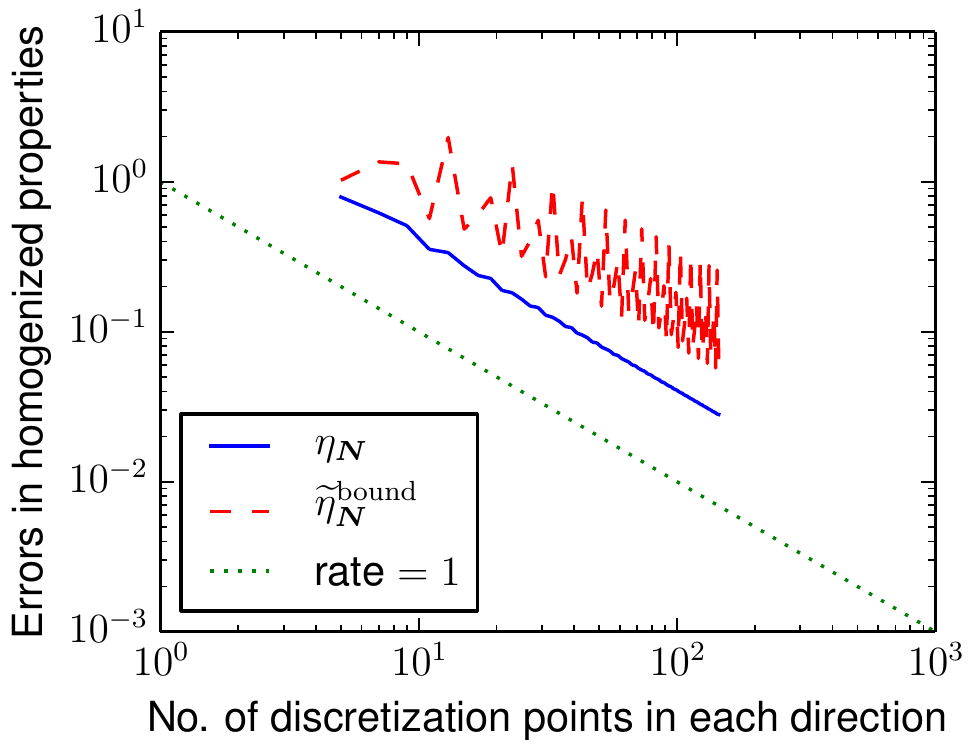}
}
\subfigure[\small Circle inclusion]{
\includegraphics[scale=0.6]{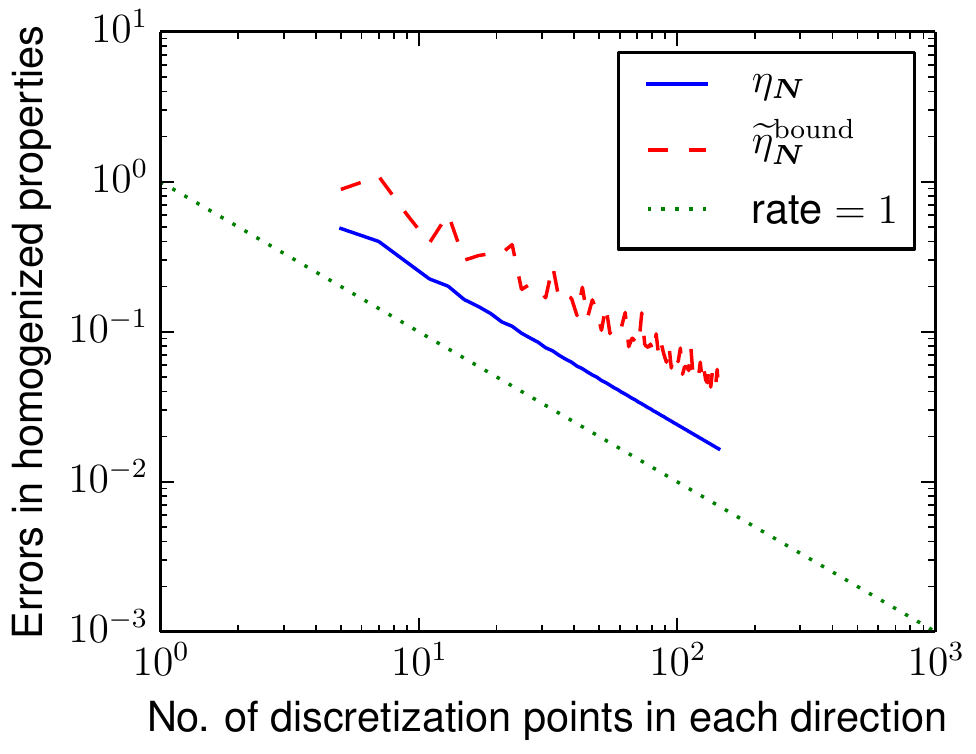}
}
\caption{Errors in homogenized properties \eqref{eq:error} for inclusion size $s=0.6$ and phase contrast $\rho=10^1$}
\label{fig:errors_for_dofs}
\end{figure}

The normalized errors, introduced and studied in \cite{Dvorak1993master} for the Finite Element Method, develop in the same rate for both schemes and this confirms the theoretical results regarding the convergence of minimizers presented in \cite[section~4.2 and 4.3]{VoZeMa2014FFTH} for FFT-based methods.
Moreover, errors in the Ga scheme evolve almost as a straight line and this allows us to predict the number of grid points required to achieve the necessary accuracy.
Finally, both the homogenized properties of the GaNi and the normalized error undergo more zigzag behavior for square than for circle inclusion, because the material coefficients change at all grid points along the square interface for a change in its size or in the number of grid points.

\subsection{Comparison of Ga with GaNi for an increase in phase ratio}
\label{sec:numerical_errors_for_phase_contrast}
This section investigates the homogenized properties in terms of normalized errors \eqref{eq:error} for an increase in phase contrast $\rho$ (see Figure~\ref{fig:errors_wrt_phase}).
Moreover, it enables fair comparion of Ga with GaNi in terms of computational and memory requirements along with the accuracy of individual methods.
Indeed, the Ga \eqref{eq:Ga_linsys_orig} and the GaNi \eqref{eq:GaNi_linsys} linear systems possess the same structure with block-diagonal matrices of material coefficients;
however, the Ga is evaluated on a double grid, resulting in higher computational and memory requirements for the same approximation order $\VN$; see Remark~\ref{rem:mem_req} for a detailed discussion.
Because of this, the GaNi is calculated with a double order $\VN$ than the Ga scheme; for this choice, the computational demands are approximately the same, while the memory requirements are even slightly lower for Ga, especially when the reduced version \eqref{eq:Ga_linsys_R} is used.

Independently of inclusion shapes, the Ga \eqref{def:Ga} progresses with sharply better rates than the GaNi \eqref{eq:GaNi}.
Moreover, for the same computational demands, the Ga scheme produces tighter guaranteed bounds on homogenized properties, which is amplified for higher phase contrasts.

\begin{figure}[htp]
\centering
\subfigure[\small Square inclusion]{
\includegraphics[scale=0.6]{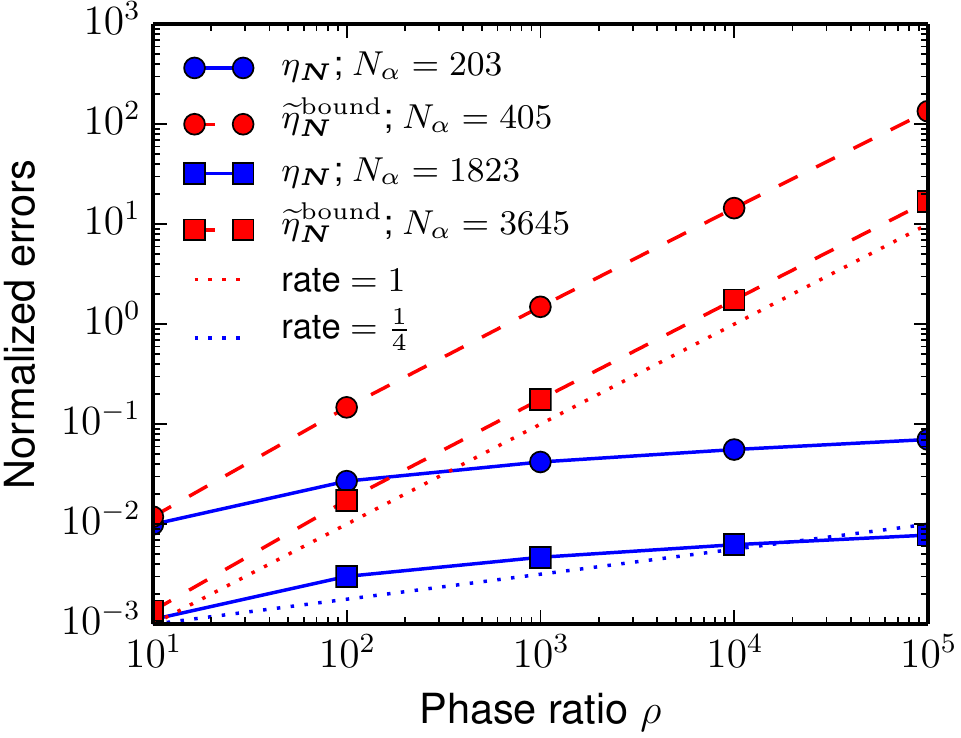}
}
\subfigure[\small Circle inclusion]{
\includegraphics[scale=0.6]{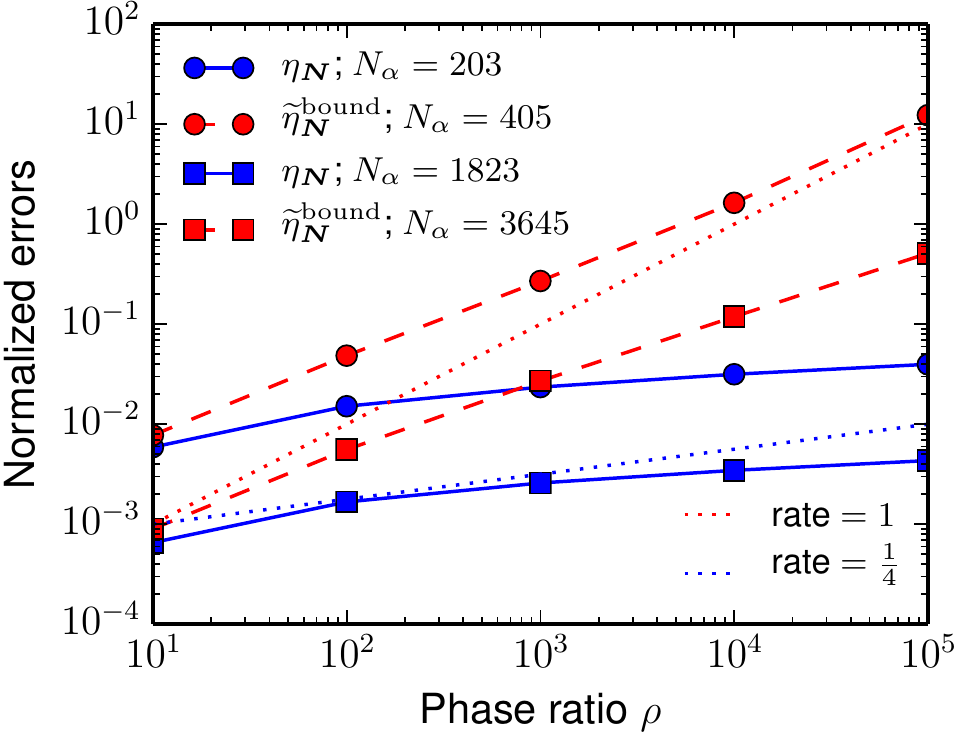}
}
\caption{Normalized errors \eqref{eq:error} for an increase in phase ratio $\rho$, $s=0.6$}
\label{fig:errors_wrt_phase}
\end{figure}

\subsection{The evolution of guaranteed bounds during iterations of conjugate gradients}
\label{sec:numerical_CG}

Here, the author investigates the evolution of bounds  during iterations of conjugate gradients (CG). In each iteration, a guaranteed bound is evaluated using the corresponding quadratic form as in \eqref{eq:minim_krylov}.
The results are shown in Figure~\ref{fig:CG} for primal formulation (upper bound), both topologies, and a relatively high phase contrast $\rho=10^4$ to highlight the behavior.

\begin{figure}[htp]
\centering
\subfigure[\small Square inclusion]{
\includegraphics[scale=0.6]{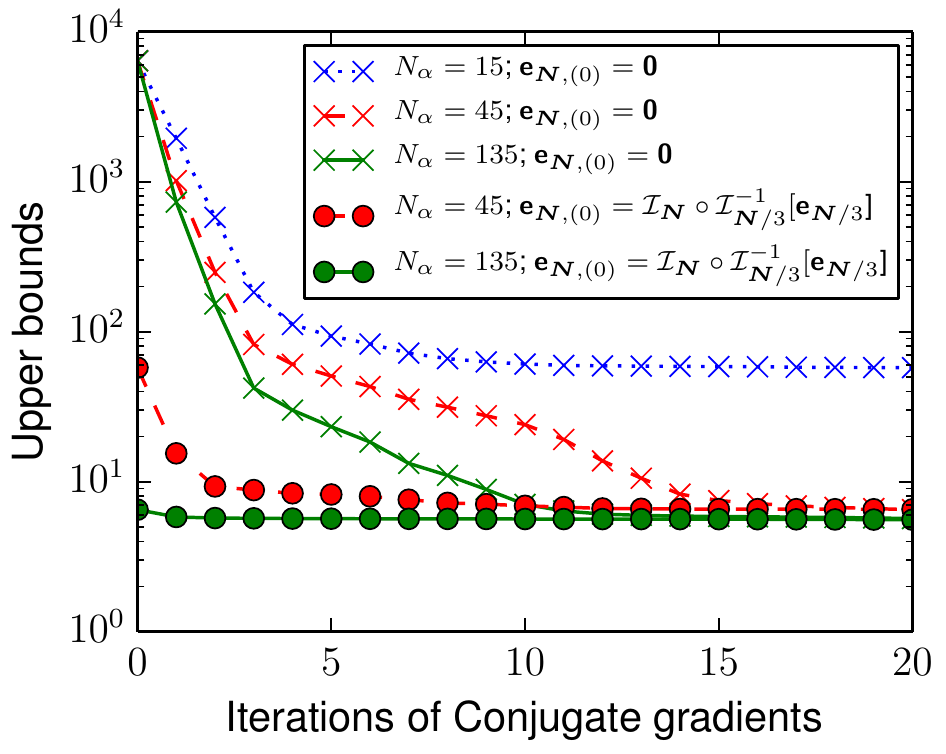}
}
\subfigure[\small Circle inclusion]{
\includegraphics[scale=0.6]{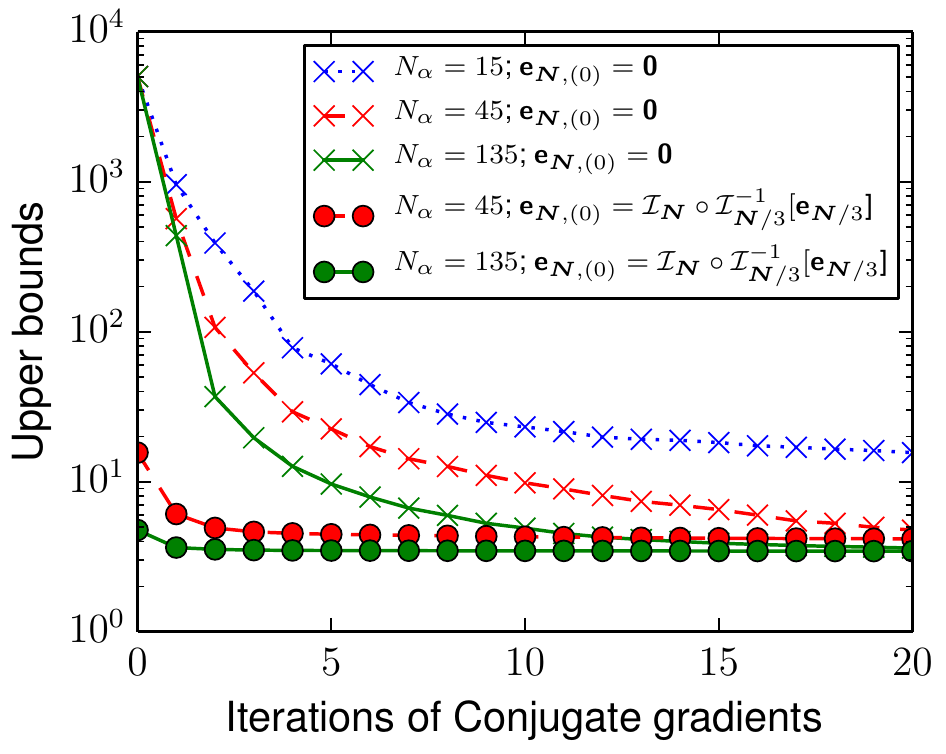}
}
\caption{Upper bounds during iterations of conjugate gradients; $\rho = 10^4$; $s=0.8$}
\label{fig:CG}
\end{figure}

According to the standard results summarized in Remark~\ref{rem:CG}, CG minimize the quadratic functional corresponding to the upper bound; the monotonic evolution of homogenized properties is confirmed in Figure~\ref{fig:CG}.
For all grid sizes, since the initial approximation for CG is taken as a zero vector, the bounds begin from a Voigt bound $\mean{\TA}$, the mean of material coefficients.

This starting point can be significantly improved using a hierarchy of approximation spaces \eqref{eq:subspace_hierarchy} in accordance to the p-version of the FEM \cite{Dvorak1993master} when a solution on a coarse grid is used as an initial approximation on a fine grid. This idea was also used for FFT-based homogenization in \cite{Eyre1999FNS}, where the prolongation was defined on nested grids with the help of modified bi-cubic Hermite polynomials. 
Here, the prolongation operator $\IM{\VN} \circ \IMi{\VM}:\xRdM\rightarrow\xRdN$ is defined on general grids $\VN,\VM\in\xNd$, $N_\alp>M_\alp$ using discretization operator \eqref{eq:IN}.
The initial approximation $\MBe\mac{\VE}_{\VN,(0)}=\IM{\VN} \circ \IMi{\VM}[\MBe\mac{\VE}_{\VM}]$ on a fine grid is then calculated from the solution of a linear system on a coarse grid $\MBe\mac{\VE}_{\VM}\in\xRdM$ with an FFT of size $\VM$ and an inverse FFT of size $\VN$;
in the case of Figure~\ref{fig:CG}, coarse grid $\VM$ is chosen to be $\VN/3$.
Note that no approximation is made in this step because the corresponding trigonometric polynomial on the coarse grid equals the one on the fine grid, i.e. $\IMi{\VM}[\MBe\mac{\VE}_{\VM}] = \IMi{\VN}[\MBe\mac{\VE}_{\VN,(0)}]\in\cE_{\VM}\subset\cEN$.
\subsection{Fly ash foam}
Here, the author shows how these methods can be applied to a  complex material consisting of alkali-activated fly ash foam. The coefficients, according to \cite{Hlavacek2014flyash},
\begin{align*}
\TA(\Vx) = \bigl[ 0.49 \cdot f(\Vx) + 0.029\cdot\bigl(1 - f(\Vx)\bigr) \bigr] \cdot \mI
\quad\text{for }\Vx\in\puc,
\end{align*}
are defined via a fly ash phase characteristic function $f:\puc\rightarrow\xR$ depicted in Figure~\ref{fig:3d_cell} as a voxel-based image with resolution $\VN=[99, 99, 99]$ corresponding to $970 299$ points.
\begin{figure}[htp]
\centering
\includegraphics[scale=0.6]{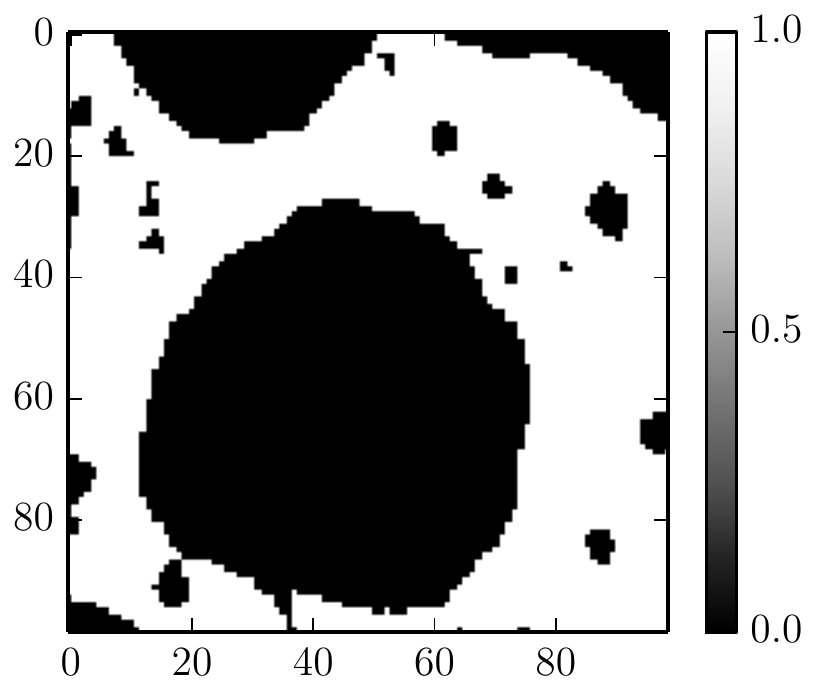}
\caption{A frontal view on a three-dimensional cell}
\label{fig:3d_cell}
\end{figure}

The models were calculated on a conventional
PC (Intel$\textcopyright$ Core$\texttrademark$ i7-4790 CPU @ 3.60GHz and $32$ GB of RAM) within less than half an hour for both the GaNi and the Ga schemes. The results
are represented for eigenvalues of homogenized coefficients because they also satisfy the  structure of upper-lower bounds \eqref{eq:structure_homog_matrices}, i.e.
for the Ga scheme \eqref{eq:Ga}
\begin{subequations}
\label{eq:3dcell_ga_eig}
\begin{align}
\eig\AeffN &= 
\begin{bmatrix}
0.12910304 & 0.13832553 & 0.14775427
\end{bmatrix},
\\
\eig\BeffN^{-1} &=
\begin{bmatrix}
0.11659856 & 0.12501264 & 0.13370292
\end{bmatrix},
\end{align}
\end{subequations}
for the GaNi scheme \eqref{eq:GaNi}
\begin{subequations}
\label{eq:3dcell_gani}
\begin{align}
\eig\tAeffN &= 
\begin{bmatrix}
 0.12635922 &  0.13525791 & 0.14487843
\end{bmatrix},
\\
\eig(\tBeffN^{-1}) &=
\begin{bmatrix}
 0.12635914 & 0.13525783 & 0.14487836
\end{bmatrix},
\end{align}
\end{subequations}
and for their corresponding guaranteed bounds \eqref{eq:GaNi_apost}
\begin{subequations}
\label{eq:3dcell_gani_apost}
\begin{align}
\eig\tAeffN\bound &= 
\begin{bmatrix}
0.13140661 & 0.14087040 & 0.15016215
\end{bmatrix},
\\
\eig\bigl(\tBeffN\bound\bigr)^{-1} &=
\begin{bmatrix}
0.10289562 & 0.11066014 & 0.11795426
\end{bmatrix}.
\end{align}
\end{subequations}
The eigenvalues of the GaNi formulation \eqref{eq:3dcell_gani} differ only because of an algebraic error and this confirms the duality of the GaNi scheme stated in \cite[Propositin~34]{VoZeMa2014GarBounds} (see Remark~\ref{rem:duality_gani} for an overview).
Moreover, they are located between the guaranteed bounds obtained by both the Ga \eqref{eq:3dcell_ga_eig} and the GaNi \eqref{eq:3dcell_gani_apost} schemes, and thus the GaNi provides an applicable prediction of homogenized properties. 
Because the guaranteed bounds comply with the energetic norms of minimizers, the Ga \eqref{eq:3dcell_ga_eig} signifies a better approximation of local fields than the GaNi \eqref{eq:3dcell_gani_apost}. This gap is accentuated in highly-contrasted media.

\section{Conclusion}
This paper focuses on the numerical solution to the variational form of the unit cell problem \eqref{eq:homog_problem}, describing the homogenized properties of periodic heterogeneous materials. For discretization, two Fourier-Galerkin schemes were used and studied: Galerkin approximation (Ga) in \eqref{eq:Ga} and its version with numerical integration (GaNi) in \eqref{eq:GaNi}. In \cite{VoZeMa2014GarBounds}, the computable guaranteed bounds on homogenized properties were introduced for the latter scheme. The approach, consisting in an exact evaluation of the primal-dual variational formulation for materials with an analytical expression of Fourier coefficients, is generalized here and applied to the Ga scheme, also resulting in a comparison with the GaNi. Theoretical results are confirmed with numerical examples. To summarize the most important findings:

\begin{itemize}
\item The structure of the guaranteed bounds on homogenized properties, originating from Ga and GaNi, was established in Proposition~\ref{lem:struct_bounds}, section~\ref{sec:bounds}.
 \item
In Lemma~\ref{lem:grid_compos}, section~\ref{sec:grid_composite}, the methodology for efficient double grid quadrature from the author's previous work \cite[section~6]{VoZeMa2014GarBounds} is generalized for a grid-based composite \eqref{eq:gridcomp}.
These material coefficients, defined via high-resolution images assuming e.g. piece-wise constant or bilinear approximation, can be effectively treated using FFT, see~Remark~\ref{rem:eval_with_FFT}.
\item Both the Ga \eqref{eq:Ga} and GaNi \eqref{eq:GaNi} schemes lead to discrete formulations with a very similar block-sparse structure; compare Remark~\ref{rem:FD_GaNi} with Lemma~\ref{lem:integ_eval_VM} and linear system \eqref{eq:GaNi_linsys} with \eqref{eq:Ga_linsys}. 
 However, the Ga is primarily evaluated on a double grid which can be recast to the original grid using shifts of DFT, see Lemma~\ref{lem:reduce_grid}.
The memory and computational requirements discussed in Remark~\ref{rem:mem_req} are higher for the linear systems of Ga \eqref{eq:Ga_linsys} than the GaNi \eqref{eq:GaNi_linsys}.
Nevertheless, the recast Ga \eqref{eq:Ga_linsys_R} leads to reduced memory requirements compared to the original Ga \eqref{eq:Ga_linsys_orig} without impacting computational
costs involved in solving linear systems.
 \item  The Ga scheme \eqref{eq:Ga} outperforms the GaNi \eqref{eq:GaNi}. Under matching computational costs for both schemes, the Ga provides more accurate results; guaranteed bounds on homogenized properties are more tight. The gap between the two schemes is accentuated in highly-contrasted media, \JV{section~\ref{sec:numerical_errors_for_phase_contrast}}.
 \item Evaluation of guaranteed bounds depends on knowing the Fourier coefficients of material properties $\TA$, see~\cite{VoZeMa2014GarBounds}. So, the approximation of $\TA$ is proposed as a way to produce upper-upper and lower-lower guaranteed bounds, section \ref{sec:approx_bounds}.
 \item Contrary to GaNi, the Ga scheme exhibits monotonous behavior, without oscillations in homogenized properties, for an increase in grid points and for a change in inclusion size, \JV{sections~\ref{sec:numerical_sensitivity2size} and~\ref{sec:numerical_bounds_wrt_dofs}}.
 \item Both schemes have the same rate of convergence of both minimizers and homogenized properties, which confirms the theoretical results in \cite{VoZeMa2014FFTH}. From the rate of convergence, it is possible to predict the grid size for a required level of accuracy, \JV{ section~\ref{sec:numerical_bounds_wrt_dofs}}.
 \item The Ga scheme can be effectively solved using conjugate gradients providing monotonous improvements of guaranteed bounds during iterations. Moreover, an approximate solution on a coarse grid can be easily transferred to a fine grid to significantly improve the convergence of the solution to the linear system, \JV{sections~\ref{sec:linear_system} and~\ref{sec:numerical_CG}}.
\end{itemize}
To conclude, I recommend using the Ga scheme because it leads to more accurate approximations for the same computational effort. Moreover, the numerical behavior of the Ga is more smooth and predictable than the GaNi.

\JV{The methodology used here is also valid for linearized elasticity. When using engineering notation (e.g. Mandel's notation) in topological dimension $3$, elasticity corresponds to a scalar problem treated here for dimension $d=6$ along with a different projection operator $\MBhG$.} Nevertheless, additional investigation is required for more complex problems.


\subsection*{Acknowledgement}
This work has been supported by project EXLIZ -- CZ.1.07/2.3.00/30.0013 which is co-financed by the European Social Fund and the national budget of the Czech Republic and by the Czech Science Foundation through project No.~P105/12/0331.

\end{document}